\documentclass[11pt,a4paper]{article}

\usepackage{amsmath, amssymb, comment}
\usepackage{amsthm}

\usepackage{framed}
\usepackage{algorithm}
\usepackage{algpseudocode}
\usepackage{stmaryrd}

\usepackage{amssymb,amsmath,xspace,enumerate,gensymb}
\usepackage{ifmtarg,color}

\usepackage{tikz}
\usetikzlibrary{arrows,decorations.pathmorphing,backgrounds,positioning,fit,petri,backgrounds}
\usetikzlibrary{calc}
\usetikzlibrary{scopes}
\usetikzlibrary{patterns}
\usepgflibrary{shapes.geometric}
\usetikzlibrary{matrix}
\usetikzlibrary{decorations.pathreplacing}

\usepackage{hyperref}
\usepackage{extarrows}
\usepackage{multirow,booktabs}

\newcommand{\progress}[2][]{}  

\clearpage{}

 \newtheorem{theorem}{Theorem}
 \newtheorem{lemma}[theorem]{Lemma}
 \newtheorem{corollary}[theorem]{Corollary}

 \newtheorem{proposition}[theorem]{Proposition}

\newenvironment{proofof}[1]{\vspace{2mm}\noindent\emph{Proof (of #1).}\enspace}{\qed\vspace{2mm}}
\newenvironment{proofsketch}{\noindent\emph{Proof sketch.}}{\qed}
    
\theoremstyle{definition}
\newtheorem{example}{Example}

\providecommand {\calA}      {{\mathcal A}\xspace}
\providecommand {\calB}      {{\mathcal B}\xspace}

\providecommand {\calD}      {{\mathcal D}\xspace}

\providecommand {\calF}      {{\mathcal F}\xspace}

\providecommand {\calI}      {{\mathcal I}\xspace}
\providecommand {\calJ}      {{\mathcal J}\xspace}

\providecommand {\calO}      {{\mathcal O}\xspace}
\providecommand {\calP}      {{\mathcal P}\xspace}

\providecommand {\calR}      {{\mathcal R}\xspace}
\providecommand {\calS}      {{\mathcal S}\xspace}

\newcommand{\bigO}{\ensuremath{\calO}}

\newcommand{\N}{\ensuremath{\mathbb{N}}}

\newcommand{\Z}{\ensuremath{\mathbb{Z}}}

\newcommand{\num}[1]{\ensuremath{[#1]}}
\newcommand{\numz}[1]{\ensuremath{[#1]_0}}

\newcommand{\NL}{\mbox{{\sf NL}}}

\newcommand{\ACz}{\mbox{{\sf AC}$^0$}}

\newcommand{\TC}{\mbox{{\sf TC}}}

\newcommand{\NCz}{\mbox{{\sf NC}$^0$}}

\newcommand{\DynTCz}{\mbox{{\sf DynTC$^0$}}}
\newcommand{\DynACz}{\mbox{{\sf DynAC$^0$}}}

\newcommand{\DynFO}{\mbox{{\sf DynFO}}\xspace}
\newcommand{\DynQF}{\mbox{{\sf DynQF}}\xspace}
\newcommand{\DynProp}{\mbox{{\sf DynProp}}\xspace}

\newcommand{\DynFOar}{\mbox{{\sf DynFO$(+,\!\times\!)$}}\xspace}
\newcommand{\DynFOarithmetic}{\DynFOar}

\newcommand{\SVrank}{\myproblem{SVRank}\xspace}
\newcommand{\Prank}{\myproblem{RankModp}\xspace}

\newcommand{\Reach}{\myproblem{Reach}\xspace}
\newcommand{\PMatching}{\myproblem{PerfectMatching}\xspace}
\newcommand{\MMatching}{\myproblem{MaxMatching}\xspace}
\newcommand{\TwoSat}{\myproblem{$2$-Sat}\xspace}

\newcommand{\reach}{\textsc{Reach}\xspace}

\newcommand{\rank}{\ensuremath{\text{rank}}}

\newcommand{\columnWidth}{10cm}

\providecommand{\nc}{\newcommand}

\newcommand{\mtext}[1]{\textsc{#1}}

\newcommand  {\myclass} [1]  {\ensuremath{\textsc{#1}}}

\nc{\dynac}{\ensuremath{\myclass{DynAC}^0}\xspace}
\nc{\adom}{\ensuremath{\text{adom}}}
\nc{\dom}{\ensuremath{\text{dom}}}
\nc{\subs}{\subseteq}

  \newcommand{\insertion}[2]{\textbf{insert}\ #1\ \textbf{into}\ #2}
  \newcommand{\deletion}[2]{\textbf{delete}\ #1\ \textbf{from}\ #2}

  \newcommand{\insertRule}[4]{
      \vspace{2mm}
      \hspace{1cm}\begin{tabular}{l}
          \textbf{on insert}\ #1\ \textbf{into}\ #2 \\
          \textbf{update}\  #3\  \textbf{as}\  #4
      \end{tabular}
      \vspace{2mm}
  }

\newcommand{\db}{\ensuremath{\calD}\xspace}

\newcommand{\inp}{\calI}
\newcommand{\aux}{\calA}

\newcommand  {\myproblem} [1] {\textsc{#1}}

\newcommand{\df}{\ensuremath{\mathrel{\smash{\stackrel{\scriptscriptstyle{
    \text{def}}}{=}}}} \;}

\newcommand{\struc}{\calS}
\newcommand{\state}{\struc}

\newcommand{\schema}{\tau}

\newcommand{\inpSchema}{\schema_{\text{in}}}
\newcommand{\auxSchema}{\schema_{\text{aux}}}

\newcommand{\relSchema}{\schema_{\text{rel}}}
\newcommand{\conSchema}{\schema_{\text{const}}}

\newcommand{\arity}{\ensuremath{\text{Ar}}}
\newcommand{\domain}{D}

\providecommand{\prog}{\ensuremath{\calP}\xspace}

\newcommand{\del}{\mtext{del}\xspace}

\makeatletter \newcommand{\uf}[4]{
  \@ifmtarg{#4}{
    \ensuremath{\phi^{#1}_{#2}(#3)}
   }{
    \ensuremath{\phi^{#1}_{#2}(#3; #4)}
  }
}

\newcommand{\acomment}[2]{\ \\ \fbox{\parbox{\linewidth}{{\sc #1}: #2}}}
\newcommand{\mcomment}[2]{{\color{blue}(#1)}\footnote{#1: #2}}

\nc{\sdm}[1]{\mcomment{SD}{#1}}
\nc{\rkm}[1]{\mcomment{RK}{#1}}
\nc{\amm}[1]{\mcomment{AM}{#1}}
\nc{\tsm}[1]{\mcomment{TS}{#1}}
\nc{\tzm}[1]{\mcomment{TZ}{#1}}

\nc{\sd}[1]{\acomment{SD}{#1}}
\nc{\rka}[1]{\acomment{RK}{#1}}
\nc{\am}[1]{\acomment{AM}{#1}}
\nc{\ts}[1]{\acomment{TS}{#1}}
\nc{\tz}[1]{\acomment{TZ}{#1}}

\newcommand{\problemIndent}{\hspace{5mm}}
\newcommand  {\querydescr} [3] {
    \vspace{3mm}
    \def\Name{#1}
    \def\Input{#2}
    \def\Question{#3}
     \begin{minipage}{3cm}
     \problemIndent\begin{tabular}{r p{\columnWidth}r}      \textit{Problem:} & \mbox{{\sf \Name}}\\
      \textit{Input:} & \Input \\
      \textit{Output:} & \Question
     \end{tabular}
     \end{minipage}

    \vspace{3mm}
    }

\newcommand  {\bquerydescr} [3] {
    \vspace{3mm}
    \def\Name{#1}
    \def\Input{#2}
    \def\Question{#3}
    \begin{minipage}{3cm}
      \problemIndent\begin{tabular}{r p{\columnWidth}r}      \textit{Problem:} & \mbox{{\sf \Name}}\\
                      \textit{Input:} & \Input \\
                      \textit{Question:} & \Question
                    \end{tabular}
                  \end{minipage}
    \vspace{3mm}
    }

 \newcommand{\mhat}[1]{\widehat{#1}}

\newcommand{\Eh}{\mhat{E}}

\newcommand{\restrict}{\upharpoonright}
\clearpage{}

    \makeatletter
    \let\@fnsymbol\@arabic
    \makeatother

    \usepackage[affil-it]{authblk}
\author[1]{Samir Datta}
\author[2]{Raghav Kulkarni}
\author[1]{Anish Mukherjee}
\author[3]{Thomas Schwentick}
\author[3]{Thomas Zeume}
\affil[1]{Chennai Mathematical Institute, India \newline
          (sdatta,anish)@cmi.ac.in}
\affil[2]{Center for Quantum Technologies, Singapore \newline 
          kulraghav@gmail.com}
\affil[3]{TU Dortmund University\newline
(thomas.schwentick,thomas.zeume)@tu-dortmund.de}

\begin{document}
\sloppy

\title{
Reachability is in DynFO
}

\maketitle

\begin{abstract}
Patnaik and Immerman introduced the dynamic complexity class \DynFO of database queries
that can be maintained by first-order dynamic programs with the help of auxiliary
relations under insertions and deletions of edges
\cite{PatnaikI97}. This article confirms their conjecture that the
Reachability query is in \DynFO. 

As a byproduct it is shown that the rank of a matrix with small values
can be maintained
in \DynFOar. It is further shown that the (size of the)
maximum matching of a graph can be maintained in non-uniform \DynFO,
another extension of \DynFO, with non-uniform initialisation of the
auxiliary relations.
\end{abstract}

\pagenumbering{arabic} 

\section{Introduction}
\progress{30}
\label{sec:intro}

In many data management scenarios, data is subject to frequent
change. When a web server is
temporarily not available, data packages have to be rerouted
immediately; when a train is cancelled on short notice, travellers need to
find alternative connections as fast as possible.

Recomputation of a query after each small change
of the data is often not possible due to the large
amount of data at hand and efficiency considerations. Often it is
also not necessary: the loss of one server usually affects only a
small part of the network. Therefore it makes sense to consider
incremental algorithms that use previously computed auxiliary data
to answer queries faster, after a small change.

In this article, we do not study the dynamic scenario from the
point of view of incremental algorithms, but rather  from the point of
view of Descriptive Complexity (see \cite{Ibook}). More precisely, we
use the setting of Dynamic Complexity Theory as introduced by Patnaik
and Immerman   \cite{PatnaikI97}. Dynamic Complexity Theory has its roots in theoretical investigations of the view update problem
for relational databases.
In a nutshell, it investigates the logical complexity 
of updating the result of a query under deletion or insertion of  tuples into a
database. 

Besides possibly saving resources, a dynamic approach to query
answering can increase the expressivity of database query
languages. It is well-known that the relational algebra
 inherits the well-known
expressivity limitations from first-order logic. It basically can only
express local queries that do not count (see, e.g., \cite{Libkin04} for more
information on the limits of first-order logic on finite structures).
In the dynamic setting, first-order logic is more powerful:
as a simple example, whether the size of a set is odd or
even can be easily maintained under single insertion and deletion
operations  with the help of a single bit of auxiliary (stored) data.

Starting with work by Dong, Su, and Topor \cite{DongT92, DongS93}
(with the name first-order incremental evaluation systems, FOIES) as
well as Patnaik and Immerman \cite{PatnaikI97}, the power of
first-order logic as an update
mechanism has been studied over the last decades. In this line of
work, the result of a query is updated by first-order formulas that
have access to the current database and to an auxiliary database that
may contain helpful information. The relations of the auxiliary
database are updated by first-order formulas as well. Beyond
the expressive equivalence with the relational algebra, first-order
logic is also an interesting update language thanks to its
correspondence to low level circuit-based complexity classes: queries
maintainable by first-order updates can also be maintained 
by highly parallel algorithms. We refer to the
class of queries that can be updated by first-order formulas under
single tuple insertions and deletions by
$\DynFO$ as introduced in \cite{PatnaikI97}.

The reachability query returns, for a given graph $G$, all pairs
$(s,t)$ of nodes, for which there is a path from $s$ to $t$.
When investigating the expressive power of \DynFO,
the reachability query is of particular interest, since it is one of
the simplest queries that can not be expressed (statically) in
first-order logic, but rather requires recursion. Actually, it is
in a sense prototypical due to its connection to transitive closure
logic. The question whether the reachability query can be
maintained by first-order update formulas has been considered as one
of the central open questions in Dynamic Complexity. It has been studied for several restricted graph classes and variants of $\DynFO$ \cite{DattaHK14,DongS98,GraedelS12,Hesse03a,PatnaikI97,ZeumeS15}. In
this article, we confirm the conjecture of Patnaik and Immerman
\cite{PatnaikI97} that the reachability query for general directed
graphs is indeed in~\DynFO. 

\begin{theorem}
\label{thm:reach}
Reachability is in \DynFO. 
\end{theorem}

The proof of Theorem~\ref{thm:reach} relies on a known reduction from
Reachability to the matrix rank query and a dynamic program for a
suitable restriction of the latter. This query, to which we refer
as \SVrank (short for \emph{rank for small valued matrices}, cf.\ Section~\ref{sec:preliminaries}), is defined on
quadratic matrices with integer values which are bounded
by the number of rows of the matrix.

More precisely, whether there is a path from a node $s$ to a
node $t$ in a graph $G$ can be reduced to the question whether a
certain matrix has maximal rank. Technically, the reduction thus
yields a collection of matrices, one for each pair $(s,t)$ of
nodes. It has the additional property that a single change in $G$
(deletion or insertion of an edge) only yields a single change in each
of these matrices. Furthermore, the reduction can use arithmetic in a
certain generic way. We formalise such reductions as \emph{bounded
  expansion first-order truth-table reductions} with arithmetic
(bfo$(+,\times)$-tt reductions). We refer to their arithmetic-free
version as bfo-tt reductions and show that \DynFO is closed under
bfo-tt reductions.

We show the following result.

\begin{theorem}
\label{thm:rank}
\SVrank is in \DynFO.
\end{theorem}

Due to the use of arithmetic in the reduction from the reachability
query to \SVrank, the immediate implication of Theorem~\ref{thm:rank}
and the reduction is that the reachability query is in \DynFOar, the
extension of \DynFO, in which dynamic programs can use an addition and
a multiplication relation on their domain from the very beginning of
the computation (unlike \DynFO programs, where the database/graph and
the auxiliary relations do not contain any tuples initially). However, we
show that with respect to  weakly domain independent queries (cf.\
Section~\ref{sec:preliminaries}), such as the reachability query, \DynFOar
and \DynFO coincide. We also show that for weakly domain independent queries the existence
of a \DynFOar program implies the existence of a FOIES.

By bfo-tt reductions to the reachability query it is not hard to show that satisfiability of 2-CNF
formulas and regular path queries for graph databases can also be
maintained in \DynFO. 

Finally, we show that two queries that deal with matchings in graphs
can be maintained in a non-uniform extension of \DynFO.

\begin{theorem}
\label{thm:matching}
  \PMatching and \MMatching are in non-uniform $\DynFO$.
\end{theorem}

Here, \PMatching is the Boolean query that asks for the existence of a
perfect matching in a graph, whereas \MMatching returns the size of
the maximum matching (encoded as a singleton unary relation). Non-uniform
\DynFO is the extension of \DynFO whose programs can use arbitrarily pre-defined
auxiliary relations (similarly to non-uniform circuit families).

\paragraph{Related work}

All steps of the proof of Theorem~\ref{thm:reach} benefit from previous work.
The underlying correspondence between graph reachability and the inverse of
its adjacency matrix has been already used long ago  \cite{Cook85} and
the precise relationship between reachability and matrix rank used in
our proof has been already stated in \cite{Laubner11}. The algorithm
constructed for maintaining the rank of a matrix adapts a dynamic
sequential algorithm from \cite{FrandsenF09}.\footnote{We note that the
  conference version of this article exhibits a different algorithm
  for this problem.}  The third step extends the technique for maintaining arithmetic from \cite{Etessami98}.

Whether the reachability query can be maintained by first-order update
programs has been one of the main questions in
Dynamic Complexity. The positive results towards
the resolution of this question can be clustered in two groups: 
\begin{enumerate}[(1)]
\item results that show how to maintain \reach in
  extensions of \DynFO, and
\item results that show how to maintain \reach in \DynFO on
  restricted classes of graphs.
\end{enumerate}
In a sense, the former line of research has won this race, since the
methods developed there led to the maintenance algorithm presented here. 

The first result in group (1) was that, on arbitrary directed graphs, Reachability  can be
maintained in (uniform) \DynTCz\  \cite{Hesse03a}. The technique was based on
generating functions for representing the number of paths of a given
length from one node to another and on the observation that only
paths up to length $n$ have to be considered.  
In \cite{DattaHK14} this approach was extended to show that
Reachability can be even maintained in non-uniform $\DynACz[2]$. In terms of logic, $\DynTCz$ and $\DynACz[2]$ can be seen as the extensions of \DynFO
in which update formulas are allowed to use majority quantifiers and
modulo-2 counting quantifiers, respectively. In the latter paper also
the dynamic complexity of matrix rank was studied for the first time
(putting it in uniform \DynTCz).

Results of group (2) showed that the reachability query can be
maintained in \DynFO for undirected graphs \cite{PatnaikI94},
directed acyclic graphs \cite{DongS93}, and embedded planar graphs
\cite{DattaHK14}. For undirected graphs, reachability can even be
maintained in \DynQF (i.e., with quantifier-free formulas using auxiliary functions)
and for acyclic deterministic graphs even in \DynProp (i.e., with quantifier-free formulas
with auxiliary relations) \cite{Hesse03}.

In the case of undirected graphs, spanning trees \cite{PatnaikI94} or distance functions \cite{GraedelS12} can be
used. The result for directed acyclic graphs is based on a smart observation that allows to figure out whether
there is a path from $a$ to $b$ after deleting some edge $(c,d)$ \cite{DongS93}.

In a third line of research on Reachability, inexpressibility results have been obtained for fragments of $\DynFO$ \cite{DongS98,DongLW03,ZeumeS15}. 

As for maximum matching, in \cite{Sankowski04,Sankowski07} a reduction from maximum matching to matrix rank  has been used to construct a dynamic algorithm for maximum matching. While in this construction the inverse of the input matrix is maintained using Schwartz Zippel Lemma, we use the Isolation Lemma of Mulmuley, Vazirani and Vazirani's \cite{MulmuleyVV87} to construct a non-uniform $\DynFO$-algorithm for maximum matching.

\paragraph{Organisation}
In Section~\ref{sec:preliminaries} we fix our notation for databases,
queries and linear algebra. In Section~\ref{sec:framework} we define
the dynamic complexity framework and introduce bfo-tt
reductions. Section~\ref{sec:reachindynfo} is dedicated to the proof
that the reachability query is in \DynFO. That regular path queries
and 2-SAT are also in \DynFO is shown in Section~\ref{sec:applications}. Section~\ref{sec:matchings}
presents our results on graph matchings.
That the reachability query can also be maintained by FOIES is shown
in Section~\ref{sec:relatedsettings}. The conclusion is given in Section~\ref{sec:conclusion}.

\section{Preliminaries}
\progress{81}
\label{sec:preliminaries}
By $\num{n}$ we denote the set $\{1, \ldots, n\}$ and by $\numz{n}$ the set $\{0,1, \ldots, n\}$. 

In this article we are interested in the following algorithmic problems:

\querydescr{\reach}{A directed graph $G$}{Set of all pairs $(u,v)$, for which there is a path from $u$ to $v$ in $G$}

\querydescr{\MMatching}{An undirected graph $G$}{The size $k$ of a maximum matching of $G$}

\bquerydescr{\PMatching}{An undirected graph $G$}{Does $G$ have a perfect matching?}

A \emph{matching} $M$ of an undirected graph $G$ is a subset of pairwise non-adjacent edges of $G$. A node is matched by $M$ if it is the endpoint of one of the edges in $M$. A \emph{maximum matching} of $G$ (also: maximum-cardinality matching) is a matching that has the largest number of edges. A \emph{perfect matching} of $G$ is a matching that matches all vertices.

\subsection{Databases and Queries}

As much of the original motivation for the investigation of dynamic
complexity came from incremental view maintenance (cf.\ \cite{DongT92,DongS93,PatnaikI97}),
it is common to consider logical structures as relational
databases and to use notation from relational databases.

A \emph{(relational) schema} $\schema$ consists of a set $\relSchema$
of relation symbols, accompanied by an arity function $\arity:
\relSchema \rightarrow \N$, and a set $\conSchema$ of constant
symbols. In this work, a \emph{domain} is a finite set. A
\emph{database} $\db$ over schema $\schema$ with domain $\domain$
assigns to every relation symbol $R \in \relSchema$ a relation of
arity~$\arity(R)$ over $\domain$  and to every constant symbol $c \in
\conSchema$ an element (called \emph{constant}) \mbox{from
  $\domain$}. The \emph{active domain} $\adom(\db)$ of a database
$\db$ consists of those
elements of $\domain$ that either occur in some relation or as a
constant. 

A  $\schema$-\emph{structure} $\struc$ is a pair $(\domain, \db)$
where $\domain$ is a domain and $\db$ is a database with domain
$\domain$ over schema $\schema$. By $\dom(\struc)$ we refer to $\domain$.
For a relation symbol $R \in \schema$ and a constant symbol~$c \in \schema$ we denote by $R^\state$ and $c^\state$ the relation and constant, respectively, that are assigned to those symbols in~$\state$. 

The distinction between structures and databases will be relevant in
the dynamic complexity framework, since there the domain $\domain$
will be static, whereas the database and its active domain might
change. However, we often do not keep the two formalisms too much
apart and, e.g.,  refer by ``database'' to the corresponding
structure, in cases where the domain is given by the context (or not important).

Often structures come with special arithmetic relations $<$, $+$, and
$\times$ that are interpreted by a linear order on the domain, its
induced addition and its induced multiplication relation. When a
linear order is present, we often identify the elements of $\domain$
with the numbers in $\{0,\ldots,|\domain| -1\}$. A pair $(a, b) \in \domain
\times \domain$ then represents the number $a\times |\domain| + b$,
where in the latter term $a$ denotes the number with which the element
$a$ is identified. Likewise for tuples of higher arity. It is well known that from
$<$, $+$, and $\times$, arithmetic for tuples can be defined in
first-order logic.

A \emph{$k$-ary query} $q$ on $\schema$-structures is a mapping that is closed under isomorphisms and assigns a subset of $\domain^k$ to every $\schema$-structure over \mbox{domain $\domain$}. The problems \reach, \MMatching and \PMatching can be represented as binary, unary and boolean queries, respectively, on graph structures, i.e.~$\{E\}$-structures where $E$ is a binary relation. For example the query representing $\reach$ maps a graph structure to a binary relation that contains the transitive closure of the graph. The query representing \MMatching assumes arithmetic relations to be present and maps a graph structure to the set $\{m\}$ where $m$ is the size of a maximum matching of $G$.

We write $R\restrict A$ for the restriction of a relation $R$ to
tuples over the set $A$, and $\db \restrict A$ for the database
resulting from \db by restricting all relations to tuples over $A$.

A query $q$ is \emph{weakly domain independent}, if \mbox{$q(\db) \restrict \adom(\db) = q(\db \restrict \adom(\db))$}, for all databases
$\db$.

\subsection{Linear Algebra and Matrices}
\label{section:linearalgebra}
By $A[i,j]$ we refer to the entry in the $i$-th row and $j$-th column
of a matrix~$A$. Similarly, $x[i]$ denotes the $i$-th entry of vector $x$.  By
$e^{(n)}_i$ we denote the $n$-dimensional unit (column) vector $e$
with $e[i]=1$ and $e[j]=0$ for~\mbox{$j\not=i$}. We write $x^\top$ if
we use vector $x$ as a row vector. 

The rank and the determinant of a matrix~$A$ are denoted by $\rank(A)$
and $\det(A)$, respectively. For a prime number $p$, we denote by
$\rank_p(A)$ the rank of $A$ as a matrix over $\Z_p$ (and with entries
adjusted modulo $p$).

To the best of our knowledge, computational linear algebra problems
like matrix rank and matrix inverse have not been studied  before in dynamic complexity (with the notable exception of
Boolean matrix multiplication in~\cite{HesseI02}). Therefore, there is
no standard way of representing the matrix rank problem in the
dynamic complexity framework. The key question is how to represent the
numbers that appear in a matrix. We use a
representation that does not allow matrices with large numbers but suffices for
our applications in which matrix entries are not (much) larger than the
number of rows in the matrix. 

We say that an $(m\times m)$-matrix $A$ over $\Z$ has \emph{small values},
if for each $i,j\in\{1,\ldots,m\}$, $|A[i,j]|\le m$.  

\querydescr{\SVrank}{$(m\times m)$-matrix $A$ with small values}{Rank of $A$}

For this query, $(m\times m)$-matrices $A$  are represented by
structures as follows. The domain of the structure contains $m+1$
elements. There is a linear order $<$ that enables us to identify the
universe with $\numz{m}$. 
There are compatible $+$- and
$\times$-relations as well. Furthermore, they have two ternary relations $A_+,A_-$ to represent
the entries of $A$. That $A[i,j]=a$, for $a\in\{1,\ldots,m\}$ is represented by a triple $(i,j,a)$
in $A_+$. Similarly, if $A[i,j]=a$, for $a\in\{-m,\ldots,-1\}$, there
is a triple $(i,j,a)$ in $A_-$.  For each $i,j$, there is at most one triple $(i,j,a)$ in
$A_+\cup A_-$. If, for some $i,j$ there is no such triple $(i,j,a)$ then $A[i,j]=0$.

The query result of \SVrank is a unary relation $Q$ that is supposed
to contain a unique element $r$, the rank of $A$.

Change operations might insert a triple $(i,j,a)$ to $A_+$ or $A_-$ (in case, no
$(i,j,b)$ is there), or delete a triple. That is,
basically, single matrix entries can be set to 0 or from 0 to some
other value. However, the relations $<$, $+$, and $\times$ cannot be changed.

\section{Dynamic Complexity}
\progress{81}
\label{sec:framework}

In this section we recall the dynamic complexity framework. 

\subsection{Dynamic programs and \DynFO}

Inputs in dynamic complexity are represented as relational structures as defined in Section~\ref{sec:preliminaries}.  The domain is fixed from the beginning, but the database in the initial structure is empty. This initially empty structure is then modified by a sequence of insertions
and deletions of tuples. 

The goal of a dynamic program is to answer a given query for the database that results from any change sequence. To this end, the program can use an auxiliary 
data structure represented by an auxiliary database over the same domain. Depending on the exact setting, the auxiliary
database might be initially empty or not.

We make this more precise now, closely following the exposition in \cite{SchwentickZ16}. A dynamic program $\prog$ works on an \emph{input structure} $\inp$ over a schema $\inpSchema$ and updates an \emph{auxiliary structure} $\aux$ over a schema\footnote{To simplify the exposition, we will usually not mention schemas explicitly and always assume that all structures we talk about are compatible with respect to the schemas at hand.} $\auxSchema$. Both structures $\inp$ and $\aux$ share the same domain $D$ which does not change during a computation. We call a pair $(\inp, \aux)$ a \emph{state} and consider it as one relational structure. The relations of $\inp$ and $\aux$ are called \emph{input and auxiliary relations}, respectively. 

The input structure can be \emph{changed} by inserting or deleting a single tuple. A \emph{change operation} is thus of the form $\insertion{\vec t}{R}$ or $\deletion{\vec t}{R}$, for some tuple $\vec t$ and input relation $R$.  For a sequence $\alpha$ of change operations and an input database $\inp$, we denote the structure resulting from applying $\alpha$ to $\inp$ by $\alpha(\inp)$.

A \emph{dynamic program} has a set of update rules that specify how auxiliary relations are updated after a change. An \emph{update rule} for updating an auxiliary relation $T$ after inserting a tuple into an input relation $R$ is of the form

    \insertRule{$\vec x$}{$R$}{$T(\vec y)$}{$\varphi(\vec x, \vec y)$}

where the formula $\varphi$ is over $\inpSchema \cup \auxSchema$. We call $\varphi$ the \emph{update formula} for $T$ under insertions into $R$. 
The semantics of such an update rule is as follows.  When a tuple $\vec a$ is inserted into  input relation $R$, then the new state $\state$ of $\prog$ is obtained by inserting $\vec a$ into $R$ and by defining each auxiliary relation $T$ via $T \df \{ \vec b \mid (\inp, \aux) \models \varphi(\vec a, \vec b)\}$. Similarly for deletions. 
For a change operation $\delta$ we denote the updated state by $\prog_\delta(\state)$, and similarly for sequences of changes.
We refer to the insertion or deletion of a tuple together with the update operations applied by $\prog$ as a \emph{change step}.

The dynamic program $\prog$ \emph{maintains} a $k$-ary query $q$ if it has a $k$-ary auxiliary relation $Q$ that, after each change sequence, contains the result of $q$ on the current input database. More precisely, for each non-empty\footnote{This technical restriction ensures that we can handle, e.g., Boolean queries with a yes-result on empty structures. Alternatively, one could use an extra formula to compute the query result from the auxiliary (and input) structure.} sequence $\alpha$ of changes and each empty input structure $\inp_\emptyset$, relation $Q$ in $\prog_\alpha(\state_\emptyset)$ and $q(\alpha(\inp_\emptyset))$ coincide. Here, $\state_\emptyset=(\inp_\emptyset, \aux_\emptyset)$, where $\aux_\emptyset$ denotes the empty auxiliary structure over the domain of $\inp_\emptyset$. 

The class of queries that can be maintained by a dynamic program with update formulas from first-order logic is called $\DynFO$.

Several dynamic settings  have been studied in the literature (see
e.g. \cite{PatnaikI97,Etessami98,GraedelS12,GeladeMS12}). Here, we concentrate on the following three dynamic complexity classes:
\begin{itemize}
\item \DynFO is the class of all dynamic queries that  can be maintained by dynamic
  programs with formulas from first-order logic starting from an empty database and empty auxiliary relations.
\item $\DynFOar$ is defined as \DynFO, but the programs have three
  particular auxiliary relations that are initialised as a linear
  order and the corresponding addition and multiplication
  relations. There might be further auxiliary relations, but they are
  initially empty.
\item \emph{Non-uniform} $\DynFO$ is defined as $\DynFO$, but the auxiliary
  relations may be initialised arbitrarily.
\end{itemize}

It is well known that first-order logic with arithmetic is as powerful as uniform $\ACz$-circuits \cite{BarringtonIS90}. This correspondence naturally transfers to the dynamic setting. That is, a query can be maintained in $\DynFOar$ if and only if it can be maintained by uniform $\ACz$-circuits.

\newcommand{\para}[1]{\overrightarrow{#1}}
\newcommand{\paraB}{\ensuremath{\para{\calB}}\xspace}

\subsection{Logical truth-table reductions}
\progress{81}

In this article, we use a more general notion of reductions between queries than the bounded-expansion first-order reductions (bfo-reductions) in \cite{PatnaikI94}. They basically compare to bfo-reductions like truth-table reductions relate to many-one reductions in Complexity Theory.

The rough idea is to reduce the computation of a query $q$ on a structure $\struc$ to the computation of a query $q'$ by
\begin{enumerate}
\item defining  in a first-order fashion, from $\struc$, a collection of structures of the form
  $\calJ(\struc,\vec a)$, one for each
  tuple $\vec a$ of some arity $m$ over the domain of~$\struc$, 
\item combining all query results $q'(\calJ(\struc,\vec a))$ into a structure
  $\struc'$, and
\item defining $q(\struc)$ from $\struc$ and $\struc'$ by a first-order formula. 
\end{enumerate}

A technical complication arises from the fact that we need to allow the structures $\calJ(\struc,\vec a)$ to be defined not over the domain of $\struc$ but over some Cartesian product over this domain. Thus, we have to deal with two ``dimension parameters'': $d$ will denote the dimension of the domain of the structures $\calJ(\struc,\vec a)$ and $m$ will denote the arity of the tuples $\vec a$.

Let in the following $\sigma,\tau$ be relational schemas.

An \emph{interpretation} $\calJ$ of dimension $d$ and arity $m$ from databases with schema $\sigma$ to databases with schema $\tau$ consists of
\begin{itemize}
\item a $\sigma$-formula $\varphi_D(\vec x,\vec y)$, and
\item $\sigma$-formulas $\varphi_R(\vec x_1,\ldots,\vec x_{\arity(R)},\vec y)$, for every $R\in\tau$,
\end{itemize}
where $\vec y=y_1,\ldots,y_m$, $\vec x=x_1,\ldots,x_d$ and, for every $j$, $\vec x_j=x_{j1},\ldots,x_{jd}$.

For every (finite) $\sigma$-structure $\struc$ and each $m$-tuple $\vec a$ over the domain of~$\struc$, the interpretation $\calJ$ defines a structure $\calJ(\struc,\vec a)$ with
\begin{itemize}
\item domain
  $D^{\vec a}\df\{\vec b \in A^d \mid \calA \models \varphi_D(\vec
  b,\vec a)\}$, and 
\item relations
  \[
R^{\vec a}\df\{(\vec b_1,\ldots,\vec b_{\arity(R)})\in (D^{\vec a})^{\arity(R)} \mid
  \calA \models \varphi_i(\vec b_1,\ldots,\vec b_{\arity(R)},\vec a)\},
\]
  for every $R$.
\end{itemize}

A \emph{first-order truth-table query-to-query reduction}  (fo-tt-reduction) \mbox{$\calR=(\calJ,\varphi)$} from $q$ to $q'$ consists of an interpretation $\calJ$ and a formula $\varphi$ with free variables $x_1,\ldots,x_k$, where $k$ is the arity of $q$, which fulfils the following reduction property.

For every (finite) structure $\struc$, $q(\struc)$ is the set $\{\vec t \mid \struc' \models \varphi(\vec t)\}$, where the structure $\struc'$ with domain $\dom(\struc)$ 
is defined as follows.
 Let $d$ be the dimension of $\calJ$, $m$ its arity, $\ell$ the arity of $q'$, and $\sigma, \tau$ the schemas of $\calJ$. 
\begin{itemize}
\item $\struc'$ has all relations from $\struc$;
\item Furthermore, $\struc'$ has a relation $\hat{Q}$ of arity $m+d\ell$ that contains all tuples of the form $(\vec a, \vec s)$, where $\vec a\in U^d$ and $\vec s\in q'(\calJ(\struc,\vec a))$. In $(\vec a, \vec s)$ the $\ell$-tuple $s$ over universe $U^d$ is considered a $d\ell$-tuple in the obvious way.
\end{itemize}

We refer to $\varphi$ as the \emph{wrap-up formula} of the reduction.

In analogy to \cite{PatnaikI94}, we say that an interpretation $\calJ$ has \emph{bounded expansion} if there is a constant \emph{expansion bound} $c$ such that for all structures $\struc_1,\struc_2$ over the same domain $D$, which differ by exactly one tuple, and for every tuple $\vec a$ over $D$, the databases $\calJ(\struc_1,\vec a)$ and  $\calJ(\struc_2,\vec a)$ differ by at most $c$ tuples. A fo-tt reduction has bounded expansion if its underlying interpretation has bounded expansion. We refer to fo-tt reductions with bounded expansion as bfo-tt reductions.

\begin{example}\label{ex:bfott}
As an illustrating example we show how the well-known reduction from \TwoSat   to \Reach can be cast as a bfo-tt reduction. The Boolean query \TwoSat  asks whether a given propositional formula in $2$-CNF has a satisfying assignment. Here, a propositional formula is in $2$-CNF if it is in conjunctive normal form and each clause contains at most two literals.

Instances of \TwoSat can be represented as structures as follows. The domain of a structure representing a formula $\varphi$ is the set of variables of $\varphi$. The clauses of $\varphi$ are represented by three binary input relations $C_{TT}$, $C_{TF}$ and $C_{FF}$ such that a tuple $(x,y) \in C_{TT}$ corresponds to a clause $x \vee y$, a tuple $(x,y) \in C_{TF}$ to a clause $x \vee \neg y$, and a tuple $(x,y) \in C_{FF}$ to a clause $\neg x \vee \neg y$. 

Thus, insertion and deletion of tuples corresponds to insertion and deletion of clauses in a natural way.

Intuitively, the reduction maps a $2$-CNF-formula $\theta$ with variables $V$ to the graph $G = (V \cup \overline{V}, E)$ where $\overline{V} = \{\neg x \mid x \in V\}$ and
$E$ contains the edges $(\neg L, L')$ and $(\neg L', L)$ if $L \vee L'$ is a clause in~$\theta$. It can be easily seen that $\theta$ is satisfiable if and only if there is no
variable $x \in V$ such that there are both a path from $x$ to $\neg
x$ and a path from $\neg x$ to $x$ in $G$.

More formally, the graph $G$ will be encoded over the set of pairs over~$V$. For two variables $u\not=v$ from $V$, a pair $(x,u)$ will represent $x$ and $(x,v)$ will represent $\overline{x}$.  This can be achieved by a 2-dimensional\footnote{We emphasise that the reduction constructions one graph for each pair $(u,v)$ with $u\not= v$. For simplicity, we ignore the case $u=v$ in the following. Graphs defined with these parameters do not contribute to the success of the reduction. We likewise ignore $2$-CNF-formulas with only one variable.}   interpretation $\calJ$ of arity 2. For each pair $(u,v)$ of variables, $\calJ(\theta,u,v)$ is the graph defined as above with $u$ and $v$ indicating positive and negated literals, respectively.  Thus, the formula $\varphi_D((x_1,x_2),(y_1,y_2))$ could be chosen as $(x_2=y_1)\lor(x_2=y_2)$ allowing only pairs in the domain of $\calJ(\theta,u,v)$ whose second entry is one of the parameters given by $y_1$ and $y_2$. The formula $\varphi_E((x_{11},x_{12}),(x_{21},x_{22}),(y_1,y_2))$ can be chosen as
$$
\big((C_{TT}(x_{11},x_{21})\lor C_{TT}(x_{21},x_{11}))\land (x_{12}=y_2) \land (x_{22}=y_1)\big) \lor \cdots
$$ 
with respective subformulas for $C_{TF}$ and $C_{FF}$.
Finally the wrap-up formula $\varphi$ can be chosen as
$$
\exists u,v (u\not=v) \land \neg \exists x (\hat{Q}((x,u),(x,v),(u,v))\land \hat{Q}((x,v),(x,u),(u,v))).
$$

Since the
modification of a single clause in $\theta$ induces only two
first-order definable modifications to the edge set of each
corresponding graph, the reduction is also bounded.
\qed
\end{example}

The relevance of bfo-tt reductions for this article stems from the following property. 

\begin{proposition}\label{prop:bfott}
  \DynFO is closed under bfo-tt reductions, that is, if there is a bfo-tt reduction from a query $q$ to a query $q'$ and $q'\in\DynFO$, then $q\in\DynFO$.
\end{proposition}
\begin{proofsketch}
 Let $(\calJ,\varphi)$ be a bfo-tt reduction from $q$ to $q'$ with expansion bound $c$. Let $\sigma,\tau$ be the schemas of $q$ and $q'$, respectively, and let $d$ be the dimension and $m$ the arity of $\calJ$. Let $\calP'$ be a dynamic program for $q'$. 
 
 The program $\calP'$ can be turned into a dynamic program $\calP$ for $\sigma$-structures that has one auxiliary relation $R$ of arity $m+d \arity(R')$, for every (input and auxiliary) relation $R'$ of $\prog'$. For each $m$-tuple $\vec a$, the program $\prog$ simulates the behaviour of $\prog'$ on  $\calJ(\struc,\vec a)$, independently. Since each change operation for $\struc$ translates into a sequence of at most $c$ change operations for $\calJ(\struc,\vec a)$, this amounts, for every tuple $\vec a$, to a sequence of at most $c$ update operations of~$\prog'$, which can be applied successively (but in parallel for different $\vec a$).

Since the query relation of $\prog'$ is one of its auxiliary relations, $\prog'$ has, in particular, the relation $\hat{Q}$ from the reduction property above available, and can therefore compute $q(\struc)$ in a first-order manner. 
\end{proofsketch}

We actually use slightly more powerful logical reductions, when we work with $\DynFOar$. We define bfo$(+,\times)$-tt reductions in almost the same way as bfo-tt reductions, but they assume in $\struc$ distinguished relations $<$, $+$ and $\times$ representing arithmetic on the universe. In such a reduction, the query $q$ must not depend on the choice of $<$, $+$ and $\times$, but $\calJ(\struc,\vec a)$ of course can. By an adaptation of the proof of Proposition~\ref{prop:bfott}, it can be shown that if there is a bfo$(+,\times)$-tt reduction from a query $q$ to a query $q'$, then $q\in\DynFOar$ if $q'\in\DynFO$. Similar reductions can be defined for non-uniform $\DynFO$ as well.

\section{Reachability is in DynFO}                                  \progress{50}
\label{sec:reachindynfo}
The goal of this section is to prove the main result of this article.

\begin{theorem}\label{theo:reachindynfo}
  $\reach\in\DynFO$.
\end{theorem}

The proof of Theorem~\ref{theo:reachindynfo} consists of four
relatively simple steps and it involves one additional query, \Prank, that will
be defined later.
\begin{enumerate}[(1)]
\item \Reach can be reduced to \SVrank by a bfo$(+,\times)$-tt reduction.
 \item \SVrank  can be reduced to \Prank by a bfo-tt reduction.
 \item  $\Prank\in\DynFO$.
\item For every weakly domain independent query $q$, if $q\in\DynFOarithmetic$
  then  $q\in\DynFO$. 
\end{enumerate}

From (1)-(3) it follows that $\Reach\in\DynFOar$. From (4)
we can conclude that $\Reach\in\DynFO$, since \Reach is
  weakly domain independent.

All four steps build to some extent on previous work. 
The basic correspondence between graph reachability and the inverse of
its adjacency matrix has been already used long ago  \cite{Cook85} and
the precise relationship between reachability and matrix rank used in
our proof has been already stated in \cite{Laubner11}. The algorithm constructed for the third step adapts
a dynamic sequential algorithm for maintaining rank from
\cite{FrandsenF09}. The last step extends the technique for
maintaining arithmetic presented in \cite{Etessami98}. 

In the following we describe the four steps separately and largely self-contained.

\subsection{From Reachability to Matrix Rank}

\label{sec:reduction}

Towards the reduction from \Reach to \SVrank, let $G$ be a graph with $n$ vertices and $A_G$ its adjacency
matrix, and let $s,t$ be vertices of~$G$. The important observation (which can be found, e.g., in
\cite[Theorem 6.1.10.]{HornJ12}) is that $I - \frac{1}{n}A_G$ is invertible and 
\[ (I - \frac{1}{n}A_G)^{-1} = I +
\sum_{i=1}^\infty{(\frac{1}{n}A_G)^i}.
\]

From the sum on the right hand side it can be easily concluded that
the inverse of $I
- \frac{1}{n}A_G$ has a non-zero entry at position $(s,t)$ if and only if
$t$ is reachable from $s$. 

For technical reasons, we prefer to deal with integer matrices and
therefore rather work with the matrix $B\df nI-A_G$, which is also
invertible.

Then the following chain of equivalences holds.
\begin{align*}
  \label{eq:1}
  \text{$t$} & \text{ is reachable from $s$ } \\
&\iff
                                                        (B^{-1})[s,t]\not=0\\
 &\iff (B^{-1} e^{(n)}_t)[s]\not=0\\
 &\iff \text{$x[s]\not=0$ holds for the vector $x=B^{-1} e^{(n)}_t$}\\
 &\iff \text{the equation $Bx=e^{(n)}_t$ has \emph{no} solution vector $x$ with
          $x[s]=0$}\\
 &\iff \text{the system
          $\begin{array}{rcl}
            Bx&=&e^{(n)}_t\\(e^{(n)}_s)^\top x & = & 0
          \end{array}$
has no solution vector $x$ at all}\\
&\iff \text{$e^{(n+1)}_t$ is not in the column space
of $B^{+{s}}$}\\
&\iff \text{$B^{+{st}}$ has rank $n+1$}
\end{align*}

Here, $B^{+{s}}$ denotes the \mbox{$((n+1)\times n)$}-matrix that is
obtained from $B$ by adding an
additional row $(e^{(n)}_s)^\top$, and $B^{+{st}}$ denotes the
extension of $B^{+{s}}$ by the additional column vector~$e^{(n+1)}_t$.

The latter equivalence holds since $B$ is invertible, and
thus $B$ and  $B^{+{s}}$ have rank $n$.

We next describe, how the above equivalence gives rise to a
bfo$(+,\times)$-tt reduction from \Reach to \SVrank. To this end we observe that, for graphs with $n$ vertices the resulting
($(n+1)\times(n+1)$-matrix $B^{+st}$ has small values, in the sense
defined in Section~\ref{sec:preliminaries}.  

It is easy to check that, in the presence of arithmetic,
$B^{+{st}}$ can be obtained from $G$ by a 2-dimensional and binary  bfo$(+,\times)$-tt
reduction $(\calJ,\varphi)$. For each database $\db$ (representing a graph $G$), and each pair $(s,t)$ over the universe $U$ of
$\db$, $\calJ(\db,(s,t))$ is a database that encodes $B^{+{st}}$.  The interpretation $\calJ$ uses two dimensions because the universe
representing $B^{+{st}}$ is of size $n+2$ for graphs with $n$
vertices.

For every pair $(s,t)$, the result relation
$\SVrank(\calJ(\db,(s,t)))$ is the set $\{(r)\}$, where $r$ is the
rank of $B^{+{st}}$. Therefore, the relation $\hat{Q}$ consists of all
triples $(s,t,r)$, for which $r$ is the
rank of $B^{+{st}}$.
Thus, for each $(s,t)$, the wrap-up formula $\varphi$ only needs to check whether $r=n+1$, where
$n$ is the number of nodes of $G$.

Finally, each change in $G$ results in only one change in $B^{+{st}}$,
for every~$(s,t)$, and therefore the reduction actually has expansion
bound 1.

\subsection{From Rank to Rank mod $p$}

\label{sec:rank2rankp}
Even though input matrices for $\SVrank$ have small entries, the maintenance algorithm on which the \DynFO-program for $\SVrank$ will be based needs to deal with matrices that have large entries. To avoid the complications that arise from the need to compute with large numbers, we show next that, in order to maintain the rank of a matrix $A$ with small values, it suffices to maintain its rank over the field $\Z_p$, for sufficiently many primes~$p$. We denote this rank by $\rank_p(A)$. Formally, this gives rise to a reduction from $\SVrank$ to the following query\footnote{As a dynamic problem, we do \emph{not} consider operations that change $p$. However, \Prank can also be maintained under these operations by simply always maintaining the rank over $\Z_p$, for every prime $p\le m^2$. }.

\querydescr{\Prank}{$(m\times m)$-matrix $A$ with values from
  $\{0,\ldots,p-1\}$, \newline prime $p\le m^2$}{Rank of $A$ over $\Z_p$}

The bound $m^2$ might appear a bit arbitrary, but we will see that it just suffices. 

The reduction from \SVrank  to \Prank is actually pretty simple.
It is based on the fact that, for large enough $m$, it suffices to maintain the rank for all primes $p \leq m^2$, as we will argue next. 
These primes indeed suffice thanks to the simple observation that $\rank(A) \geq k$ if and only if $\rank_p(A) \geq k$ for some prime $p \leq m^2$. Clearly, $\rank(A) \geq k$ if and only if there is some $k \times k$ submatrix $A'$ of $A$ with $\det(A')\not=0$. Since $\det(A')$ is bounded by $m!m^m$, its binary representation has $\bigO(m\log m)$ digits (for sufficiently large~$m$). Since, for large enough $n$,  there are more than $\frac{n}{\ln n}$ prime numbers between 1 and $n$ \cite{RosserB62}, there are more than $\frac{m^2}{2\ln m}$ prime numbers below $m^2$ and thus their product is clearly larger than $m!m^m$. Therefore, \mbox{$\det(A')\not=0$} if and only if there exists a prime  $p\le m^2$ such that ${\det(A') \not\equiv 0} (\,\text{mod}\, p)$. Thus, for large enough $m$, $\rank(A) \geq k$ if and only if there exists a prime  $p\le m^2$ such that $\rank_p(A) \geq k$.

The actual reduction $(\calJ,\varphi)$ from \SVrank  to \Prank is 1-dimensional and binary. Each pair $\vec i= (i_1,i_2)$ over $m$ is interpreted as the number $n(\vec i)\df (i_1-1)m+(i_2-1)$. For a database \db representing an input matrix $A$ for \SVrank and each pair $\vec i= (i_1,i_2)$, $B^{\vec i}\df \calJ(\db,\vec i)$ is the all-zero matrix if $n(\vec i)$ is not a prime. If $n(\vec i)$ is some prime $p$, then $\calB^{\vec i}$ represents the matrix $A$ over $\Z_p$.  Whether a number $n(\vec i)$ is a prime number can be tested by a first-order
formula thanks to the availability of $\times$. 

The wrap-up formula $\varphi$ simply computes the maximum $k$, such that for some prime $p=n(\vec i)$, the result relation for $\Prank(\calJ(\db,\vec i))$ contains~$k$.

\subsection{Maintaining Rank mod $p$ in DynFO}
\label{sec:rank}

In this subsection, we show that the rank of a matrix modulo some prime $p$ can be maintained in $\DynFO$. This is the most important intermediate result for Theorem
\ref{theo:reachindynfo} and interesting also in its own right. First
we give an informal description of the dynamic algorithm for matrix
rank.  Afterwards we describe how it can be transformed into a \DynFO
program. In the following we fix $m>0$,  a prime $p<m^2$, and only
consider matrices with small values.

The algorithm is an adaptation of a dynamic algorithm that has been stated in \cite{FrandsenF09}.
The idea is to maintain an invertible matrix $B$ and a matrix $E$ in
reduced row-echelon form such that $BA=E$. That $E$ is in reduced
row-echelon form means that
\begin{itemize}
\item the left-most non-zero entry (the \emph{leading entry}) in every row is 1,
\item the column of such a leading entry only contains zero-entries otherwise, and
\item rows are sorted in a ``diagonal'' fashion, that is, the column
  numbers of leading entries strictly increase with the row number.
\end{itemize}
 
The rank of $A$ equals the number of non-zero rows of $E$ thanks to $\rank(E)=\rank(B A)=\rank(A)$ and the structure of $E$. Thus maintaining the matrices $B$ and $E$ suffices to maintain the rank of $A$.
  
We describe next, how those matrices can be maintained after a
change of $A[i,j]$, for any $i,j\le m$.  All computations of matrix
entries are modulo~$p$. Let $A'$
denote the new value of matrix $A$ after this change. We explain,
how new matrices $B'$ and $E'$ can be obtained such that $B'A'=E'$.

After a change of $A[i,j]$, the product
$BA'$ differs from $BA$ at most in column $j$. Thus, to get the
desired matrix $E'$ in reduced echelon form, we can proceed as
follows. 

\begin{enumerate}[(1)]
  \item If column $j$ has more than one  leading entry  of $BA'$:
    \begin{itemize}
    \item let some entry with a maximum number of successive zeros in
      its row (right after column $j$) be the new leading entry,
    \item set this leading entry to 1, and set all other entries of column $j$ to~$0$, by appropriate
      row operations.
    \end{itemize}
 \item If a former leading entry of a row $k$ is lost in column $j$ (by the change in $A$ or
    by step (1)),
    \begin{itemize}
    \item set its new leading entry (i.e., the next non-zero entry in
      row $k$ and some column $\ell>j$) to 1 and  set all other
      entries of column $\ell$ to $0$, by appropriate
      row operations.\footnote{Since all other columns with  leading
        entries have only one non-zero entry, and row $k$ has no
        non-zero entries before column $\ell$, these row operations do
        not do any harm to the echelon structure of the rest of the matrix.}
   \end{itemize}
 \item If needed: move the (at most two) rows, for which the position
   of the leading entry has
   changed (compared with $E$)  to their correct (row)
    positions.
\end{enumerate}
An illustrating example can be found  in Figure~\ref*{fig:rank}. 
\begin{figure}[h!]
\centering\footnotesize

\begin{tabular}[c]{cccccc}
& $B$ & & $A$ & & $E$\\
& \begin{tikzpicture}[baseline]
\matrix (u) [matrix of nodes,left delimiter=(,right delimiter=)] 
{
4 & 0 & 0 & 0 & 0\\
0 & 3 & 0 & 0 & 0\\
4 & 0 & 1 & 0 & 0 \\
0 & 2 & 0 & 3 & 0 \\
3 & 0 & 0 & 0 & 1 \\
};
\end{tikzpicture}
&$\times$&
\begin{tikzpicture}[baseline]
\matrix (a) [matrix of nodes,left delimiter=(,right delimiter=)] 
{
4 & 0 & 3 & 0 & 0 \\
0 & 2 & 4 & 0 & 0 \\
4 & 0 & 3 & 1 & 0\\
0 & 2 & 4 & 0 & 2 \\
3 & 0 & 1 & 0 & 0 \\
};
\end{tikzpicture}
&$=$&
\begin{tikzpicture}[baseline]

\matrix (a) [matrix of nodes,left delimiter=(,right delimiter=)] 
{
1 & 0 & 2 & 0 & 0 \\
0 & 1 & 2 & 0 & 0 \\
0 & 0 & 0 & 1 & 0 \\
0 & 0 & 0 & 0 & 1 \\
0 & 0 & 0 & 0 & 0 \\
};
\end{tikzpicture}\\
\midrule

& $B$ & & $A'$ & & $B\times A'$\\
\begin{tikzpicture}[baseline]
\draw[->] (-2,-0.8) to [bend left = 20] node[left=-.5mm] {\tiny$+2\cdot$}  (-2, -0.0);
\draw[->] (-3,-0.8) to [bend left = 20] node[left=-.5mm] {\tiny$+3\cdot$}  (-3, 0.4);
\draw[->] (-4,-0.8) to [bend left = 20] node[left=-.5mm] {\tiny$+2\cdot$}  (-4, 0.8);
\node (tmp) at (-1.7,-.8){\tiny$\cdot 2$};
\end{tikzpicture}\hspace{-5mm}
& 
\begin{tikzpicture}[baseline]
\matrix (u) [matrix of nodes,left delimiter=(,right delimiter=)] 
{
4 & 0 & 0 & 0 & 0\\
0 & 3 & 0 & 0 & 0\\
4 & 0 & 1 & 0 & 0 \\
0 & 2 & 0 & 3 & 0 \\
3 & 0 & 0 & 0 & 1 \\
};
\end{tikzpicture}

&$\times$&
\begin{tikzpicture}[baseline]
\matrix (a) [matrix of nodes,left delimiter=(,right delimiter=)] 
{
4 & \fbox{1} & 3 & 0 & 0 \\
0 & 2 & 4 & 0 & 0 \\
4 & 0 & 3 & 1 & 0\\
0 & 2 & 4 & 0 & 2 \\
3 & 0 & 1 & 0 & 0 \\
};
\end{tikzpicture}
&$=$&
\begin{tikzpicture}[baseline]

\matrix (a) [matrix of nodes,left delimiter=(,right delimiter=)] 
{
1 & 4 & 2 & 0 & 0 \\
0 & 1 & 2 & 0 & 0 \\
0 & 4 & 0 & 1 & 0 \\
0 & 0 & 0 & 0 & 1 \\
0 & 3 & 0 & 0 & 0 \\
};
\end{tikzpicture}\\
\midrule

& $B$ after Step (1) & & $A'$ & & $B \times A'$ after Step (1)\\

\begin{tikzpicture}[baseline]
\draw[->] (-2,0.4) to [bend left = 20] node[left=-.5mm] {\tiny$+4\cdot$}  (-2, .8);
\node (tmp) at (-1.7,0.4){\tiny$\cdot 3$};
\node (unseen) at (-4,-.8) {};
\node (unseen) at (-4,.8) {};
\end{tikzpicture}\hspace{-10mm}

& 
\begin{tikzpicture}[baseline]
\matrix (u) [matrix of nodes,left delimiter=(,right delimiter=)] 
{
0 & 0 & 0 & 0 & 2\\
4 & 3 & 0 & 0 & 3\\
0 & 0 & 1 & 0 & 2 \\
0 & 2 & 0 & 3 & 0 \\
1 & 0 & 0 & 0 & 2 \\
};

\end{tikzpicture}
&$\times$&
\begin{tikzpicture}[baseline]
\matrix (a) [matrix of nodes,left delimiter=(,right delimiter=)] 
{
4 & 1 & 3 & 0 & 0 \\
0 & 2 & 4 & 0 & 0 \\
4 & 0 & 3 & 1 & 0\\
0 & 2 & 4 & 0 & 2 \\
3 & 0 & 1 & 0 & 0 \\
};
\end{tikzpicture}
&$=$&
\begin{tikzpicture}[baseline]

\matrix (a) [matrix of nodes,left delimiter=(,right delimiter=)] 
{
1 & 0 & 2 & 0 & 0 \\
0 & 0 & 2 & 0 & 0 \\
0 & 0 & 0 & 1 & 0 \\
0 & 0 & 0 & 0 & 1 \\
0 & 1 & 0 & 0 & 0 \\
};
\end{tikzpicture}\\

\midrule

& $B$ after Step (2) & & $A'$ & & $B \times A'$ after Step (2)\\

& 
\begin{tikzpicture}[baseline]
\matrix (u) [matrix of nodes,left delimiter=(,right delimiter=)] 
{
1 & 2 & 0 & 0 & 4\\
2 & 4 & 0 & 0 & 4\\
0 & 0 & 1 & 0 & 2 \\
0 & 2 & 0 & 3 & 0 \\
1 & 0 & 0 & 0 & 2 \\
};
\end{tikzpicture}
&$\times$&
\begin{tikzpicture}[baseline]
\matrix (a) [matrix of nodes,left delimiter=(,right delimiter=)] 
{
4 & 1 & 3 & 0 & 0 \\
0 & 2 & 4 & 0 & 0 \\
4 & 0 & 3 & 1 & 0\\
0 & 2 & 4 & 0 & 2 \\
3 & 0 & 1 & 0 & 0 \\
};
\end{tikzpicture}
&$=$&
\begin{tikzpicture}[baseline]

\matrix (a) [matrix of nodes,left delimiter=(,right delimiter=)] 
{
1 & 0 & 0 & 0 & 0 \\
0 & 0 & 1 & 0 & 0 \\
0 & 0 & 0 & 1 & 0 \\
0 & 0 & 0 & 0 & 1 \\
0 & 1 & 0 & 0 & 0 \\
};
\end{tikzpicture}\\

\midrule

& $B$ after Step (3) & & $A'$ & & $B \times A'$ after Step (3)\\

& 
\begin{tikzpicture}[baseline]
\matrix (u) [matrix of nodes,left delimiter=(,right delimiter=)] 
{
1 & 2 & 0 & 0 & 4\\
1 & 0 & 0 & 0 & 2 \\
 2 & 4 & 0 & 0 & 4\\
0 & 0 & 1 & 0 & 2 \\
0 & 2 & 0 & 3 & 0 \\
};
\end{tikzpicture}
&$\times$&
\begin{tikzpicture}[baseline]
\matrix (a) [matrix of nodes,left delimiter=(,right delimiter=)] 
{
4 & 1 & 3 & 0 & 0 \\
0 & 2 & 4 & 0 & 0 \\
4 & 0 & 3 & 1 & 0\\
0 & 2 & 4 & 0 & 2 \\
3 & 0 & 1 & 0 & 0 \\
};
\end{tikzpicture}
&$=$&
\begin{tikzpicture}[baseline]

\matrix (a) [matrix of nodes,left delimiter=(,right delimiter=)] 
{
1 & 0 & 0 & 0 & 0 \\
0 & 1 & 0 & 0 & 0 \\
 0 & 0 & 1 & 0 & 0 \\
0 & 0 & 0 & 1 & 0 \\
0 & 0 & 0 & 0 & 1 \\
};
\end{tikzpicture}\\

\end{tabular}
\caption{Illustration of the modifications necessary for one change in
matrix $A$ for $p = 5$. }\label{fig:rank}
\end{figure}

For each of the row operations in the algorithm, the same operation is applied to the matrix $B$. This ensures that $B' A' = E'$. As all these row operations correspond to multiplying a suitable elementary matrix from the left, $B$ remains invertible (see, e.g., \cite[p. 133]{Meyer00}). Each of the three steps can be performed in constant parallel
time.

A dynamic program can be easily obtained from the informal description given above. It maintains auxiliary relations that encode the matrices $B$ and $E$. As each of the steps (1)-(3) can be performed in constant parallel time and  since $<$, $+$ and $\times$ are available in the logical representation of the matrix $A$, the steps can be translated into a \DynFO update program~$\calP$ in a straight-forward way.

\begin{theorem}\label{prop:fmr:in:dynfoarithmetic}
  $\Prank$ is in $\DynFO$.
\end{theorem}

\subsection{DynFO(+, $\times$) vs. DynFO}
\progress{90}
\label{sec:domainindependent}

Since, \reach is clearly weakly domain independent, the
proof for Theorem~\ref{theo:reachindynfo} can be completed by a proof for
the following result.

\begin{proposition}\label{proposition:acisfo}
  If a query \mbox{$q\in \DynFOarithmetic$} is weakly
  domain independent, then \mbox{$q\in \DynFO$}. 
\end{proposition}

Etessami already observed that \DynFO programs have the same
expressive power as \DynFOarithmetic programs, if, before the actual change sequence starts, for
each element $u$ of the universe, the edge $(u,u)$ is inserted and
afterwards deleted \cite{Etessami98}. He described how these
preliminary changes can be used to construct a linear order and
compatible $+$ and $\times$ predicates on the whole universe.  He also
observed that, alternatively, arithmetic can be defined incrementally, so that at any point there are
relations $<_{ad}$, $+_{ad}$ and $\times_{ad}$ that represent a linear
order on the activated elements, and corresponding ternary addition and multiplication
relations, respectively. Here, an element $u$ of the domain is called
\emph{activated}
by a change sequence $\alpha=\delta_1,\ldots,\delta_\ell$, if $u$ occurs in
some $\delta_i$, no matter, whether an edge with $u$ is still present
after the whole sequence $\alpha$. In the following, we refer
to the set of activated elements of change sequence $\alpha$ by
$A(\alpha)$ and by $A$, if $\alpha$ is clear from the
context.\footnote{Usually, $\alpha$ is just the sequence of all
  changes of the computation at hand.}

We point out the subtle differences
between elements of the domain, elements of the \emph{active} domain, and activated elements. The domain contains all elements that can be
used in relations and does not change during a dynamic
computation. The active domain $\adom(\db)$ of a database $\db$ that results from a
change sequence $\alpha$ (applied to an initially empty database)
consists of all elements that occur in some tuple of $\db$. An element
is activated, if it occurs somewhere in $\alpha$. In particular, every
element of $\adom(\db)$ is activated and every activated element is in
the domain, but not necessarily vice versa. For example, adding the edges $(1,2)$ and $(2,3)$ to an initially empty graph over domain $\{1, 2, 3, 4\}$ and subsequently deleting the edge $(1, 2)$, yields an input database with active domain $\{2,3\}$ and activated elements $1,2,3$.

We show next that  \DynFO programs can simulate \DynFOarithmetic programs for weakly 
domain independent queries without any form of preprocessing.

\begin{proof}[Proof (of Proposition \ref{proposition:acisfo})]
Let $q$ be a weakly domain independent query and $\prog$ a \DynFOarithmetic program that maintains
$q$.  For simplicity, we assume that $q$ uses only one binary relation
$E$, the adaptation for arbitrary structures is straightforward.
We recall that change sequences
are applied to an initially empty structure, but that $\prog$ uses
non-empty initial relations that provide a linear order
and the corresponding addition and multiplication relations on the
full universe.

We will construct a \DynFO program $\prog'$ that simulates
$\prog$. By definition of~$\DynFO$, $\prog'$ has to maintain $q$ under
change sequences from an initially empty structure (just as $\prog$) but
with initially empty auxiliary relations (unlike $\prog$). 
The challenge is therefore that $\prog'$ cannot simply
simulate $\prog$ right from the beginning of the change sequence, as it does
not have $<$, $+$ and $\times$ available.

We first give a rough description of the construction of
$\prog'$. More details will be given below.  
The update program $\prog'$ maintains a linear
order $<$ on$~A$. Thanks to the linear
order, we can always associate $A$ with a set of size $m+1$ of the form $\numz{m}$, for some
number $m$, with the natural linear order. In fact, we assume for the
moment that $A$ is always of this form. Furthermore, an addition
relation and a multiplication relation on $A$  is maintained, just as in \cite{Etessami98}.

For the construction of $\prog'$ we view
computations of $\prog$  as a sequence of stages, based on the size of
$A$.  More precisely, we say that a
computation of $\prog$ on a universe $U$ of size $n$ is in stage
$i< \log \log n + 1$, if more than $N_i$ but at most $N_{i+1}$ elements of
$U$  are activated, where $N_i \df 2^{2^i}$, for every $i \geq 0$. We will ignore the case where $\leq N_0 = 2$ elements are activated in the following; it can be easily dealt with separately.

The basic idea of the construction of $\prog'$ is to use different \emph{threads} that
simulate the different stages of $\prog$ and we refer to the thread that is responsible
for stage $i$ as thread $i$.

For each $i$, thread $i$ begins as soon as $\prog$ enters stage $i-1$
and ends at the end of stage $i$ of $\prog$. During stage $i$ of $\prog$, thread $i$ is \emph{in
  charge}. The query result for $q$ is
always provided by the thread that is in charge. See Figure
\ref{figure:proposition:acisfo} for an illustration. 

  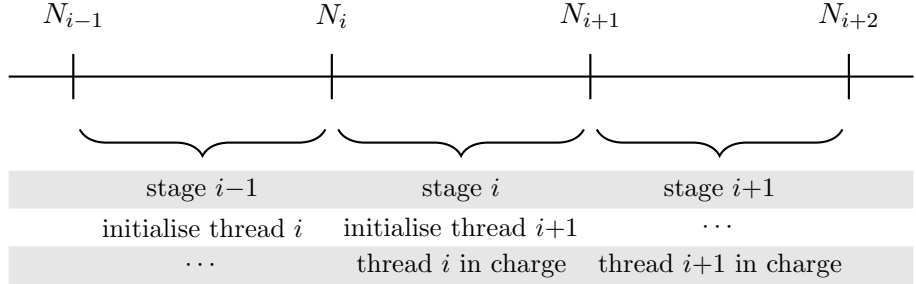
\begin{figure}[t] 
    \begin{center}
     \scalebox{1.0}{
        \begin{tikzpicture}[
             xscale=0.85,
          ]
          
          \draw [thick] (1,0) to (15,0);
          
          \draw [thick] (2,0.3) to (2,-0.3);
          \draw [thick] (6,0.3) to (6,-0.3);
          \draw [thick] (10,0.3) to (10,-0.3);
         \draw [thick] (14,0.3) to (14,-0.3);

          \node (name) at (2,0.8)  {$N_{i-1}$};
          \node (name) at (6,0.8)  {$N_{i}$};
          \node (name) at (10,0.8)  {$N_{i+1}$};
         \node (name) at (14,0.8)  {$N_{i+2}$};

          \draw [decorate,thick, align=left, decoration={brace,mirror,amplitude=10pt}](2.1,-0.7) -- (5.9,-0.7);
          \draw [decorate,thick, align=left, decoration={brace,mirror,amplitude=10pt}](6.1,-0.7) -- (9.9,-0.7);
          \draw [decorate,thick, align=left, decoration={brace,mirror,amplitude=10pt}](10.1,-0.7) -- (13.9,-0.7);

\small
          
          \draw [line width=5mm,color=black!10] (1,-1.5) to (15,-1.5);

          \draw [line width=5mm,color=black!10] (1,-2.5) to (15,-2.5);
 
         \node (name) at (4,-1.5)  {stage $i{-}1$ };
         \node (name) at (8,-1.5)  {stage $i$ };
         \node (name) at (12,-1.5)  {stage $i{+}1$ };

         \node (name) at (4,-2)  {initialise thread $i$};
        \node (name) at (8,-2)  {initialise thread $i{+}1$};
        \node (name) at (12,-2)  {$\cdots$};
  
         \node (name) at (4,-2.5)  {$\cdots$};
         \node (name) at (8,-2.5)  {thread  $i$ in charge};
          \node (name) at (12,-2.5)  {thread  $i{+}1$ in charge};
      \end{tikzpicture}
     
      }
      \caption{Illustration of the stages used in the proof of Proposition \ref{proposition:acisfo}.}
      \label{figure:proposition:acisfo}
    \end{center}
  \end{figure}

When thread $i$ starts, a linear order, an
addition relation and a multiplication relation over
$\numz{N_{i-1}}$ are available. From these relations a linear order, an
addition relation and a multiplication relation on 4-tuples over
$\numz{N_{i-1}}$ can be easily defined in first-order logic.

Thread $i$ uses the set of 
4-tuples over $\numz{N_{i-1}}$
as universe of size $N_{i+1}$. It uses one $4k$-ary auxiliary
relation~$R'$, for every (auxiliary or input)
$k$-ary relation $R$ of $\prog$. 
It starts on the structure over
$(\numz{N_{i-1}})^4$ with empty input relations and with the linear order and the
corresponding addition and multiplication relations over
$(\numz{N_{i-1}})^4$ available. It is thus in the position to simulate the behaviour  of
$\prog$ on an initially empty structure.

By $E'$, we refer to the 8-ary relation of thread $i$ corresponding to
the input relation $E$ of $\prog$.
When stage $i-1$ starts, relation $E$ might already contain 
up to (around)  $N_{i-1}^2$ edges, whereas $E'$ is
empty, since thread $i$ has not started yet. Therefore, thread $i$ can
not immediately simulate $\prog$ in a lock-step fashion, but it first
has to catch up with $\prog$. Indeed, thread $i$ will make sure
that at the end of stage $i-1$ all tuples in $E$ have corresponding
tuples in $E'$, so that it is prepared to be in charge. 

In order to catch up, thread
$i$ needs to add more than one edge per step. It is not hard to figure
out that
it suffices to add at most four edges per step to $E'$. The details will
be given below. During stage $i$, thread $i$ can simulate $\prog$
in a lock-step fashion, mimicking every step. The universe
$(\numz{N_{i-1}})^4$ is large enough to represent each new element that is
activated by some 4-tuple over $N_{i-1}$. After stage $i$, thread $i$
is abandoned and thread $i+1$ takes over. The moment, when thread $i+1$ has to take over can be recognised by maintaining a counter for each thread: if the counter of thread $i$ reaches the value $(N_{i-1})^4$ then thread $i+1$ has to take over in the next step. 
\\

\noindent
Next, we describe $\prog'$ in more detail.

We describe first, how $\prog'$ constructs a linear order\footnote{We
  use infix notation for $<$, $+$ and $\times$. } $<$, an addition
relation $+$ and a multiplication relation $\times$ on the set $A$ of activated
elements. This part of the simulation is just as
in~\cite{Etessami98}. We recall that $\prog$ 
 never changes its linear order, addition relation and
 multiplication.

The relation $<$ orders the activated elements, in the order of
activation. For concreteness: if an edge 
$(a,b)$ is inserted which activates $a$ and $b$ then $a<b$ become the
two largest elements of $<$.

The update formula for determining whether a tuple $(y_1, y_2)$ is in the relation $<$ after inserting an edge $(x_1, x_2)$ into $E$ states that 
\begin{itemize}
\item $y_1<y_2$; or
\item $y_1$ is already activated, $y_2=x_1$ or $y_2=x_2$, and is not yet activated; or
\item $y_1=x_1$, $y_2=x_2$, $y_1\not=y_2$,  and both $y_1$ and $y_2$ are not yet activated. 
\end{itemize}
That an element $x$ is activated can be expressed by $\exists y\;
x < y \lor y < x$.

For deletion operations, nothing has to be changed that is, 
$\uf{<}{\del E}{x_1,x_2}{y_1,y_2}=y_1<y_2$.

We always identify activated elements with their position in $<$,
that is, the minimal element in $<$ is considered as 0, the second
as 1 and so on. We use numbers as constants in formulas. It is
straightforward to replace these numbers by ``pure'' formulas. For
example, the subformula $x>1$ can be replaced by $\exists x_1 \exists
x_2\; x_1< x_2 \land x_2 < x$.

The formulas for $+$ and $\times$ are in the same spirit and use
the well-known inductive definitions of addition and multiplication,
respectively. 

Thread $i$ considers 4-tuples as 4-digit base-$N_{i-1}$ numbers and
thus identifies a 4-tuple  $(u_1, u_2, u_3, u_4)$ over $\numz{N_{i-1}}$  with the number $u_1\times
N^3_{i-1} + u_2\times
N^2_{i-1} + u_3\times
N_{i-1} +  u_4$.

We note that $<$, $+$, $\times$ can be lifted to
relations over 4-tuples in first-order logic and therefore, they need
not be maintained as auxiliary relations.
    
During stages $i-1$ and $i$, thread $i$ maintains a bijection $g_i$ between
the activated elements and 4-tuples over $\numz{N_{i-1}}$.  At
the start of thread $i$, $g_i(k)=(0,0,0,k)$, for
every~\mbox{$k\in\numz{N_{i-1}}$}, and $g_i$ is extended in a
straightforward fashion. 

As explained above, thread $i$ has to catch up during stage $i-1$ to
make sure that at the beginning of stage $i$, its relation $E'$ is
isomorphic to $E$ under $g_i$. To this end, thread $i$ decreases the size of
the symmetric difference $\Delta\df g_i(E) \vartriangle E'$ of
$g_i(E)$ and $E'$ by three tuples, for each change step, in the following fashion. When a change $\delta$
occurs, it is first simply applied to~$E$, without triggering the
associated update operations. Afterwards, thread $i$ identifies the
lexicographically smallest (up to) four pairs $e_1,e_2,e_3,e_4$ over
$\numz{N_{i}}$, whose image under $g_i$ is in $\Delta$ and sequentially applies the
appropriate update operations. That is, if $g_i(e_j)\in g_i(E)\setminus E'$
it simulates the operations of $\prog$ for an insertion of $g_i(e_j)$ and otherwise
for deletion. It is easy to see that, in this way, $|\Delta|$ indeed decreases by at
least three, unless $|\Delta|< 4$ already. Since $|\Delta|\le
N_{i-1}^2$ initially, $\frac{1}{3}N_{i-1}^2$ change steps would suffice for thread $i$ to
catch up. Since stage $i-1$ has at least $\frac{1}{2}(N_i-N_{i-1})$
change steps and $N_i-N_{i-1}=N_{i-1}^2-N_{i-1}\ge \frac{3}{4}N_{i-1}^2$ for
$i\ge 2$, this really works out.\footnote{The border case $i=1$ and
  $|\delta|\le 3$ can be handled in a straightforward, mostly
  analogous way.}

The query result during stage $i$ is always $g_i^{-1}(Q')$, where $Q'$
is the auxiliary relation corresponding to relation $\prog$'s query relation
$Q$ in $\prog'$.

This completes the description of the behaviour of thread $i$, for each
$i$. 

Of course, it is not possible to let each thread use its own
set of auxiliary relations. However, we can simply increase the arity
of each relation by one and use the additional entry to indicate, for
each tuple, the number of the thread, for which it is used. As an
example, all relations $E'$ are encoded
into one 9-ary relation $\Eh$ and the relation $E'$ of thread $i$ is just the
set of tuples $\{t\mid (i,t)\in \Eh\}$.

The
correctness of $\prog'$ can be shown in two steps. Let $\alpha$ denote
some change sequence and, for each $i$, let $\alpha_i$ denote the
prefix of $\alpha$ that lasts until the end of stage $i$. First, for each
$i$, it can be shown that at the start of stage $i$, the auxiliary
relations of thread $i$ are the image under $g_i$ of the auxiliary
relations of state $\prog_{\alpha'}(\state_\emptyset)$, for some
change sequence $\alpha'$ that yields the same structure as $\alpha_{i-1}$, relative to
domain $\numz{N_{i+1}}$.  Second, it
is easy to see that during stage $i$, program $\prog'$ can 
correctly keep track of the changes and updates. Thanks to the weak
domain independence of query $q$, the
output $g_i^{-1}(Q')$ is correct during stage $i$, and therefore
$\prog'$ has a correct output, at any time.
\end{proof}

We remark that Proposition \ref{proposition:acisfo} does not hold for arbitrary queries. For example, the domain dependent boolean query $q_\text{even}$, which is true for domains of even size and false otherwise, cannot be maintained in $\DynFO$ from scratch. This is because the first-order initialisation formulas cannot tell domains of even and odd size apart for large, empty structures (see, e.g., \cite{Libkin04}).

\section{Some Applications}
\progress{99}
\label{sec:applications}
From Theorem~\ref{theo:reachindynfo} the maintainability of other important queries can be inferred. The bounded first-order reduction shown in Example~\ref{ex:bfott} immediately yields 
\begin{corollary}\label{theo:rpq}
  \TwoSat is in \DynFO.
\end{corollary}

We next exhibit a straightforward bounded first-order reduction from regular path queries to Reachability.

Graph databases have received considerable attention in the database theory community recently (see e.g. \cite{MendelzonW95,AbiteboulV99, ArenasCP12, LosemannM13} and \cite{Wood12survey, Baeza13} for surveys). Usually they contain huge amounts of data, and therefore queries on graph databases should be evaluated in parallel and, if possible, dynamically. In the following we show that the answer of a (fixed) regular path query can be re-evaluated dynamically after a modification of a graph database.

We make those notions more precise first.  Graph databases are usually modeled by directed graphs with edge labels from a finite alphabet. A regular path query $q$ is a regular expression over label names. Evaluating $q$ on a graph database $G$ yields all pairs $(u,v)$ of nodes for which there is a (not necessarily simple) path from $u$ to $v$ in $G$ whose sequence of labels is in the language specified by~$q$.\footnote{The set of labels actually needs not be
  fixed a priori. However, given a regular expression $r$, only labels
  that occur in $r$ are relevant for maintaining $r$ and all other
  labels can be replaced by some fixed label $X$ not occurring in
  $r$.} We thus model graph databases as finite structures with one binary relation $E_a$ per edge label $a$.

In the following we show that regular path queries can be maintained in $\DynFO$. To this end we present a simple and well-known bounded first-order reduction to the Reachability problem \cite{KahlerW03}. Let $A$ be an NFA for the language of a regular path query $q$, let $Q$ be its set of states, and let $s_0$ be the initial state and
$s_f$ the unique accepting state of $A$, respectively. Let the synchronised product $G\times A$ of $G$ and $A$ be the (directed, unlabeled) graph with node set
$G\times Q$ and an edge from $(u_1,p_1)$ to $(u_2,p_2)$ if there is an $a$-labeled edge from $u_1$ to $u_2$ for some symbol $a$, for which there is also a transition from $p_1$ to $p_2$. Then,
$(u,v)\in q(G)$ holds if and only if $(v,s_f)$ is reachable from $(u,s_0)$ in  $G\times A$. Since
each single change in $G$ only induces at most $|Q|$ first-order definable changes in
$G\times A$, the reduction is bounded and therefore, the
maintainability of (fixed) regular path queries follows from Theorem~\ref{theo:reachindynfo}. This easily extends to conjunctions of regular path queries since $\DynFO$ is closed under boolean operations. 

\begin{corollary}\label{theo:rpq}
    Regular path queries and conjunctions thereof can be maintained in \DynFO.
\end{corollary}
Further classes of query languages for labeled graphs have been studied in the dynamic context in the literature, see \cite{WeberS07,Zeume15thesis,MunozVZ16}.

\section{Maintaining the Size of Maximum Matchings}
\progress{60}
\label{sec:matchings}

Matchings in graphs are one of the most studied graph-theoretical concepts in Computer Science with many applications (see, e.g., \cite{LovaszPP86, KarpinskiR98}). In this section, we show that the size of maximum matchings, and therefore also the existence of a perfect matching, can be maintained in non-uniform~$\DynFO$. We recall that non-uniform~$\DynFO$ programs can use initial databases that can depend on the size of the universe in an arbitrary, even non-computable way. 

To this end we reuse the techniques developed in Section \ref{sec:reachindynfo} for maintaining the rank of a matrix.  In previous work a non-uniform dynamic program with $\TC^0$-updates has been obtained for both problems \cite{DattaHK14}. It remains open whether (maximum or perfect) matching can be maintained in uniform  $\DynFO$.

The basic idea of our approach relies on a correspondence between the rank of the Tutte matrix of a graph and the size of maximum matchings. The \emph{Tutte matrix} $T_G$ of an undirected graph $G$ is the $n \times n$ matrix with entries 
$$t_{ij} =     
  \begin{cases} 
    x_{ij} & \text{if $(i,j) \in E$ and $i < j$}  \\
    -x_{ji} & \text{if $(i,j) \in E$ and $i > j$}  \\  
    0 & \text{if $(i,j) \not\in E$} 
  \end{cases}$$
where the $x_{ij}$ are indeterminates.\footnote{The rank of a matrix with indeterminates can be defined as the size of the largest quadratic submatrix with non-zero determinant.}

\begin{theorem}[Lov\'{a}sz \cite{Lovasz79}]\label{theorem:maximum-matching}
  Let $G$ be a graph with a maximum matching of size $m$. Then $\rank(T_G) = 2m$.
\end{theorem}

Unfortunately, the rank maintenance algorithm presented in Section \ref{sec:rank} cannot be applied immediately as the entries of $T_G$ are indeterminates, and applying the maintenance algorithm to a matrix with indeterminates might yield polynomials with  exponentially many terms. However, the rank of $T_G$ can be determined by computing the rank for a matrix obtained by replacing the indeterminates in $T_G$ by well-chosen positive integer values. For a graph~$G$, let $w$ be a function that assigns a positive integer weight to every edge $(i, j)$ and let $B_{G,w}$ be the integer matrix obtained from $T_G$ by substituting $x_{ij}$ by $2^{w(i,j)}$. 
\begin{theorem}\label{theorem:maximum-matching-weighted}
  If $G$ is a graph with a maximum matching of size $m$ and $w$ is a weight assignment for the edges of $G$ then $\rank(B_{G, w}) \leq 2m$. Furthermore, if $G$ has a maximum matching with unique minimal weight with respect to $w$ then~\mbox{$\rank(B_{G, w}) = 2m$}.
\end{theorem}

This theorem is implicit in Lemma 4.1 in \cite{Hoang}.
 For the sake of completeness we give a full proof here. The proof uses the following theorem.
\begin{theorem}[Mulmuley, Vazirani and Vazirani \cite{MulmuleyVV87}] \label{theorem:mulmuley}
  Let $G$ be a graph with a perfect matching and $w$ a weight
  assignment such that $G$ has a unique perfect matching with minimal weight with respect to $w$. Then~\mbox{$\det(B_{G, w}) \neq 0$}.
\end{theorem}

\begin{proofof}{Theorem \ref{theorem:maximum-matching-weighted}}
  Recall that the rank of a matrix can be defined as the size of the largest submatrix with non-zero determinant. Thus $\rank(B_{G, w}) \leq \rank(T_G)$, and therefore $\rank(B_{G, w}) \leq 2m$ by Theorem \ref{theorem:maximum-matching}. 
  
  For showing $\rank(B_{G, w}) \geq 2m$ when $G$ has a maximum matching of unique minimal weight with respect to $w$, we adapt the proof of Theorem~\ref{theorem:maximum-matching} given in \cite{RabinV89}. Let $U$ be
  the set of vertices contained in the unique maximum matching of $G$ with
  minimal weight, and $G'$ the subgraph of $G$ induced by~$U$. Observe that $G'$ has a unique minimal weight \emph{perfect} matching with respect to $w$. Restricting $B_{G, w}$ to rows and columns labeled by elements from $U$ yields the matrix $B_{G', w'}$ where $w'$ is the weighting $w$ restricted to edges from $G'$.  However, then $\det (B_{G', w'}) \neq 0$ by Theorem \ref{theorem:mulmuley} and therefore~\mbox{$\rank(B_{G, w}) \geq 2m$}.
\end{proofof}

Using the technique implicit in \cite{ReinhardtA00} one can find, for every $n \in \N$, weighting functions $w_1, \ldots, w_{n^2}$ with weights in $\num{4n}$,   such that for every graph $G$ over $\num{n}$ there is an $i \in \num{n^2}$ such that  $G$ has a maximum matching with unique minimal weight with respect to $w_i$. 

We show how to obtain those functions. The following lemma is due to Mulmuley, Vazirani and Vazirani \cite{MulmuleyVV87}, but we use the version stated in \cite[Lemma 11.5]{Jukna01}.

\begin{lemma}[Isolation Lemma] \label{lemma:isolation}
  Let $m, M \in \N$ and  let $\mathcal{F} \subseteq 2^{\num{m}}$ be a non-empty set of subsets of $\num{m}$. If a weight
  function $w\in \num{M}^{\num{m}}$ is uniformly chosen at random, then with probability at least $1- \frac{m}{M}$, the minimum weight subset in $\mathcal{F}$ is unique; where the weight of a subset $F \in \calF$ is $\sum_{i \in F} w(i)$.
\end{lemma}
\begin{lemma}[Non-uniform Isolation Lemma, implicit in \cite{ReinhardtA00}] \label{lemma:non-uniform-isolation}
  Let $m \in \N$ and \mbox{$\calF_1, \ldots, \calF_{2^m} \subseteq 2^{\num{m}}$}. There are  weight functions $w_1, \ldots, w_m$ from $\num{4m}^{\num{m}}$ such that for any $i\in\num{2^m}$ with $\calF_i\not=\emptyset$, there exists a $j \in \num{m}$  such that the minimum weight subset of $\calF_i$ with respect to $w_{j}$ is unique. 
\end{lemma}
\begin{proof}
  The proof is implicit in the proof of Lemma 2.1 in \cite{ReinhardtA00}. We give a self-contained presentation.
  
  We call a collection $u_1, \ldots, u_m$ of weight functions  \emph{bad} for some $\calF_i$ if no $F \in \calF_i$ is a minimum weight subset with respect to any $u_j$. For each $\calF_i\not=\emptyset$ the probability of a randomly chosen weight sequence $U \df u_1, \ldots, u_m$ to be bad is at most $(\frac{1}{4})^m$ thanks to Lemma \ref{lemma:isolation} (for $M \df 4m$). Thus the probability that such a $U$ is bad for \emph{some} $\calF_i$ is at most $2^m \times (\frac{1}{4})^m < 1$. Hence there exists a sequence $U$ which is good for all $\calF_i$.
\end{proof}

We immediately get the following corollary.

\begin{corollary}\label{corollary:isolation:perfect-matching}
Let $G_1,\ldots,G_{2^{n^2}}$ be some enumeration\footnote{For notational
  simplicity we use $n^2$ instead of $n \choose 2$, here.}
of the graphs on
$\num{n}$ and let $\calF_1, \ldots, \calF_{2^{n^2}}$ be their respective sets
of perfect matchings. There is a sequence $w_1, \ldots, w_{n^2}$ of weight assignments assigning a value from $\num{4n^2}$ to the edges of $\num{n}^2$ such that for every graph $G$ over $\num{n}$ there is some $i \in \num{n^2}$ such that if $G$ has a perfect matching then it also has a perfect matching with unique minimal weight with respect to $w_i$.
\end{corollary}

Then, in order to maintain the size of maximum matchings of a graph $G$ over $\num{n}$, it is sufficient to maintain $\rank(B_{G, w_i})$ for all $i \in \num{n^2}$. The rank of $T_G$ is the maximal rank among those ranks thanks to Theorem  \ref{theorem:maximum-matching-weighted}.

\begin{theorem}\label{theorem:matching}
  \PMatching and \MMatching are in non-uniform $\DynFO$.
\end{theorem}

\begin{proofsketch}
  It suffices to show that \MMatching is in non-uniform~$\DynFO$. The idea is to advise a dynamic program with the weighting functions $w_1, \ldots, w_{n^2}$ that assign weights such that for all graphs with $n$ nodes there is a maximum matching with unique weight. The advice is given to the dynamic program via the initialisation of the auxiliary relations. The program then maintains the ranks for the matrices $B_{G, w_i}$ and outputs the maximal such rank. We make this more precise in the following. 
  
  Recall that the weighting functions assign values of up to $4n^2$  , and that therefore the determinant of each $B_{G, w_i}$ can be of size up to $n!(2^{4n^2})^n \leq 2^{5n^3}$,  and thus it is sufficient to maintain the rank of those matrices modulo up to $5n^3$ many primes, which are contained in the first $n^4$ numbers by the prime number theorem\footnote{We disregard small values of $n$ as the query can be directly encoded with first-order formulas for such values.}. 
 
 The dynamic program computes, for each of the weighting functions $w_i$ and each prime $ p \leq n^4$, the rank of $B_{G, w_i}$ modulo $p$.

 This informal description can be formalised by exhibiting a non-uniform bfo-tt reduction from $\MMatching$ to $\Prank$. Such reductions are defined as bfo-tt reductions but they can assume arbitrary additional relations on the structure. By an adaptation of the proof of Proposition~\ref{prop:bfott}, it can be shown that if there is a non-uniform bfo-tt reduction from a query $q$ to a query $q'$, then $q$ is in non-uniform $\DynFO$ if $q'\in\DynFO$.
 
 The non-uniform reduction  $(\calJ, \varphi)$ from $\MMatching$ to $\Prank$ is 2-dimensional and $6$-ary. Two dimensions are used to encode graphs as matrices as described in Section \ref{section:linearalgebra}. For a parameter $(p_1, \ldots, p_6)$, the interpretation $\calJ$ maps a given graph to an instance of $\Prank$ that asks for the rank of $B_{G, w_i}$ mod $p$ where $w_i$ is encoded by the first two parameters and $p$ is encoded by the remaining four parameters. For converting edges to entries of $B_{G, w_i}$ mod $p$, the reduction uses non-uniform relations. The wrap-up formula determines the highest rank of all those instances.

\end{proofsketch}

\section{First-order Incremental Evaluation Systems}
\progress{40}
\label{sec:relatedsettings}

The dynamic complexity framework was independently formalised in the
form of \emph{first-order incremental evaluation systems} (short:
\emph{FOIES}) by Dong, Su and Topor (see \cite{DongT92, DongS93}, and
also \cite{DongS97,DongS98}). Apart from notational differences, the
$\DynFO$-setting and FOIES differ in how they treat domains. While
domains in $\DynFO$ are fixed before a dynamic computation starts, the
FOIES-framework allows for the domain to grow and shrink. More
precisely\footnote{We do not give a formal definition of FOIES, but
  only describe how they differ from $\DynFO$-programs.}, the domain of a state of a FOIES computation is the active
domain, that is, the set of elements contained in some tuple of the
input database. Thus, a tuple can be inserted that contains an element
that is not contained in the current  domain, and the domain is
extended by this element. Likewise, when a tuple is removed, and one
of its elements is not contained in any other tuple afterwards, then
that element is removed from the domain.\footnote{In some papers, FOIES may use a bounded number of elements  that are not used by an input tuple \cite{DongS97}.}

Thus in FOIES, the notions of domain, active domain and activated
elements coincide, at any point in time. Yet, FOIES have an infinite
background universe $U$ and the domain $D$ always satisfies $D \subset
U$. We show next that Reachability can also be maintained by FOIES-programs.

\begin{theorem}\label{theorem:DynFOToFOIES}
    If a query \mbox{$Q\in \DynFOarithmetic$} is domain-independent, then $Q$ can be maintained by a FOIES-program.
\end{theorem}

\begin{proofsketch}
  Let $Q$ be a domain-independent query and $\calP$ a \DynFOarithmetic program that maintains $Q$.  We assume that $Q$ uses only one binary relation~$E$, the adaptation for arbitrary structures is straightforward.

  We will construct a FOIES program $\calP'$ that simulates $\calP$. To this end we extend the construction used in the proof of Proposition \ref{proposition:acisfo}. 

  The program $\calP'$ handles changes to the (active) domain by using
  the simulation technique presented in the proof of Proposition
  \ref{proposition:acisfo}. As before, the computations of $\calP$ are
  split into stages and $\calP'$ uses one thread for simulating
  $\calP$, for each different stage. Here, the stages are based on the
  size of the (active) domain $A$. A computation of $\calP$ is in stage $i$, if more than $N_i$ but at most $N_{i+1}$ elements are contained in $A$, where $N_i \df 2^{2^i}$. The update program $\calP'$ uses one thread per stage of $\calP$, the $i$-th thread being responsible for providing the correct query result whenever more than $N_i$ but at most $N_{i+1}$ elements are in the domain.
    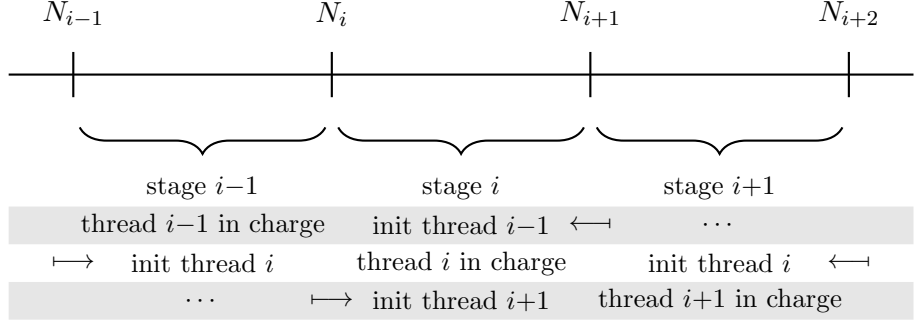
\begin{figure}[t] 
    \begin{center}
     \scalebox{1.0}{
        \begin{tikzpicture}[
             xscale=0.85,
          ]
          
          \draw [thick] (1,0) to (15,0);
          
          \draw [thick] (2,0.3) to (2,-0.3);
          \draw [thick] (6,0.3) to (6,-0.3);
          \draw [thick] (10,0.3) to (10,-0.3);
         \draw [thick] (14,0.3) to (14,-0.3);

          \node (name) at (2,0.8)  {$N_{i-1}$};
          \node (name) at (6,0.8)  {$N_{i}$};
          \node (name) at (10,0.8)  {$N_{i+1}$};
         \node (name) at (14,0.8)  {$N_{i+2}$};

          \draw [decorate,thick, align=left, decoration={brace,mirror,amplitude=10pt}](2.1,-0.7) -- (5.9,-0.7);
          \draw [decorate,thick, align=left, decoration={brace,mirror,amplitude=10pt}](6.1,-0.7) -- (9.9,-0.7);
          \draw [decorate,thick, align=left, decoration={brace,mirror,amplitude=10pt}](10.1,-0.7) -- (13.9,-0.7);

\small
          
          \draw [line width=5mm,color=black!10] (1,-2.) to (15,-2.);

          \draw [line width=5mm,color=black!10] (1,-3.) to (15,-3.);
 
         \node (name) at (4,-1.5)  {stage $i{-}1$ };
         \node (name) at (8,-1.5)  {stage $i$ };
         \node (name) at (12,-1.5)  {stage $i{+}1$ };
 
        \node (name) at (4,-2)  {thread  $i{-}1$ in charge};
        \node (name) at (8,-2)  {init thread $i{-}1$};
         \node at (10,-2) {$\longmapsfrom$};
        \node (name) at (12,-2)  {$\cdots$};

         \node at (2,-2.5) {$\longmapsto$};
         \node (name) at (4,-2.5)  {init thread $i$};
        \node (name) at (8,-2.5)  {thread  $i$ in charge};
        \node (name) at (12,-2.5)  {init thread $i$};
         \node at (14,-2.5) {$\longmapsfrom$};
  
         \node (name) at (4,-3)  {$\cdots$};
          \node at (6,-3) {$\longmapsto$};
         \node (name) at (8,-3)  {init thread $i{+}1$};
        \node (name) at (12,-3)  {thread  $i{+}1$ in charge};
     \end{tikzpicture}
     
      }
      \caption{Illustration of the stages used in the proof of Proposition \ref{proposition:acisfo}.}
      \label{figure:proposition:foies}
    \end{center}
  \end{figure}

  Two issues need to be addressed when constructing the
  FOIES-program~$\calP'$: (1) Thread $i$ uses arithmetic on $N_{i-1}$
  elements. When one of those elements is removed from the domain,
  then the arithmetic has to be adapted. (2) In contrast to the
  construction for  Proposition  \ref{proposition:acisfo}, $\calP$ can enter a stage multiple times because elements can be inserted \emph{and} removed from the domain.
  
  The first issue can be easily resolved. Using Etessami's technique, the program $\calP'$ maintains arithmetic relations $<$, $+$, and $\times$ on the domain. When the $j$th element $a$ of the domain is removed, then the arithmetic is adapted by replacing $a$ by the current maximal element $b$ with respect to $<$. To this end, first all tuples containing $b$ are removed from the auxiliary relations, and then every auxiliary tuple $\vec d$ containing $a$ is substituted by the tuple $\vec d'$ obtained by replacing all occurrences of $a$ in $\vec d$ by~$b$. Note that after those substitutions the functions $g_i$ are still bijections between 4-tuples over the first $N_i$ elements of the domain and $N_{i+2}$.
  
  We now address the second issue. Thread $i$ shall provide the correct result whenever the domain contains more than $N_i$ but at most $N_{i+1}$ elements; also when this interval has been reached by removing elements from the domain. In order to guarantee this, thread $i$ is started  either when the $(N_{i-1}+1)$-th element is inserted or when the $(N_{i+2}+1)$-th element is removed from the domain. It stops when the size of the domain is not in $\{(N_{i-1}+1), \ldots, N_{i+2}\}$.

  The way in which thread $i$ acts is just as in the proof of
  Proposition \ref{proposition:acisfo}. For each change operation
  $\delta$, in addition to performing the updates for $\delta$, the
  program $\calP'$ inserts up to five tuples to $\Eh_i$ and simulates
  $\calP$ for these five insertion steps. This guarantees that the
  query result can be decoded from the query result of thread $i$, when the domain contains $n \in \{(N_{i}+1), \ldots,
  N_{i+1}\}$ elements. This has been already shown for the case where
  thread $i$ is started by the insertion of the
  $(N_{i-1}+1)$-th element in the proof of  Proposition
  \ref{proposition:acisfo}. The argument for the case where the
  thread is started when the $(N_{i+2}+1)$st element is removed
  from the domain is similar. At least $\frac{N_{i+2}-N_{i+1}}{2} =
  \frac{N_{i+1}^2 - N_{i+1}}{2}$ edge deletions are needed to get from
  an domain of size $N_{i+2}$ to an domain of size $N_{i+1}$. Thus,
  when arriving at an  domain of size $N_{i+1}$, up to $2(N_{i+1}^2 -
  N_{i+1})$ edges are contained in $\Eh_i$ since each deletion
  contributes four edges to $\Eh_i$ as long as edges are missing in
  $\Eh_i$ (only four, as one edge may be removed from $\Eh_i$ by
  $\delta$). Since there are at most $N^2_{i+1} \geq 2(N_{i+1}^2 -
  N_{i+1})$ edges in a graph over an domain of size $N_{i+1}$, the
  edge transfer is completed before thread $i$ is in charge.

\end{proofsketch}

\section{Conclusion and Future Work}
\progress{60}
\label{sec:conclusion}

In this article we showed that Reachability can be maintained in
$\DynFO$ and thereby confirmed the conjecture of Patnaik and Immerman
\cite{PatnaikI97}. The proof adapts and combines several known
techniques in a surprisingly elementary way. One of the key
ingredients, the maintainability of the rank of a matrix with
first-order update formulas, is of independent interest. As an
immediate consequence of those results, regular path queries and
2-satisfiability can also be maintained in $\DynFO$. By combining the
linear algebraic part of the proof that Reachability is in $\DynFO$
with the Isolation Lemma, we showed how the size of a maximum matching
can be maintained in $\DynFO$ with non-uniform initialisation. 

Reachability is arguably one of the most important algorithmic
problems in Computer Science, and algorithms for solving Reachability
are the basis for solving many other problems. For this reason, the
open status of the maintainability question for Reachability has
stifled progress in the study of descriptive dynamic complexity
severely. The positive answer to this question raises hopes that other
areas become accessible to the methods of dynamic complexity now. One
example being query languages for graph databases, as is illustrated
by the dynamic program for maintaining regular path queries. Other
potential candidate areas are dynamic model checking and query
evaluation under ontologies. However, the \DynFO bound for Reachability 
 does not extend to all of \NL, simply because \DynFO is not known to
 be closed under \emph{unbounded} first-order 
 reductions. 

Also the basic techniques used for maintaining Reachability are promising further progress. Linear algebraic problems, such as the Rank problem, have thus far been neglected in the study of dynamic complexity. Also techniques known from the study of small static complexity classes have not been systematically tested in the dynamic framework. The application of the Isolation Lemma presented here indicates that this might be worthwhile. 

We plan to explore those techniques by trying to apply them to related problems such as maintaining a reachability witness, the (shortest) distance, the number of paths, whether there is a matching (and witnesses for that), the value of the determinant, and disjoint paths.

Another interesting direction for future research is to explore whether the dynamic programs for maintaining Reachability, Rank, and the size of a maximum matching presented here can be generalised and optimised. We indicate some intriguing challenges in the following.

In this article, we only looked at modifications of a single tuple. A
closer analysis reveals that the dynamic program for Rank still works
when whole columns can be modified, and therefore Reachability can be
maintained when all incoming edges of one node can be modified at once
(dually: all outgoing edges)\footnote{This observation is due to
  William Hesse.}. Whether Reachability and Rank
can be maintained for other, more complex modifications remains
open. We note that this is closely related to the question which
fragments of transitive closure logic can be maintained by first-order
updates.  

It would also be worthwhile to study whether Reachability (as well
as the other problems studied here) can be maintained in fragments of
$\DynFO$. Typical fragments limit the arity of the auxiliary relations
or the syntactic shape of update formulas.  The dynamic programs
presented here have very high arity, which makes them hard to apply in
practical scenarios. It remains open whether Reachability can be
maintained with auxiliary relations of small arity. So far it is
only known that Reachability cannot be maintained using unary
auxiliary relations. Another interesting question is whether
Reachability can be maintained by even weaker update mechanisms,
e.g. \NCz-updates. Lower bounds for this fragment are conceivable;
yet, even for the quantifier-free fragment of $\DynFO$, which
corresponds to restricted $\NCz$-updates, lower bounds are
nontrivial. It is only known that  binary auxiliary relations are not
sufficient to maintain Reachability in this fragment of $\DynFO$
\cite{ZeumeS15}.

\section{Acknowledgements}
\progress{60}
We thank William Hesse for stimulating and illuminating discussions. We are grateful to an anonymous reviewer for pointing us to the dynamic algorithm for maintaining rank by Frandsen and Frandsen \cite{FrandsenF09}. This simplified the presentation enormously. Further we thank Nils Vortmeier for many very helpful comments on drafts of this article. The first and the third authors were partially funded by a grant from Infosys Foundation. The last two authors acknowledge the financial support by DFG grant \mbox{SCHW 678/6-1}

   \bibliographystyle{plain}
  \bibliography{dyn}

\end{document}